\documentclass[a4paper]{article}
\usepackage[utf8]{inputenc}
\usepackage{verbatim}
\usepackage{tcolorbox}
\usepackage{algorithm}
\usepackage{algpseudocode}
\usepackage{amsmath,amsfonts,amsthm,graphicx,authblk,tikz,enumerate}
\usepackage[T1]{fontenc}
\usepackage[utf8]{inputenc}
\usepackage{authblk}
\usepackage[a4paper,top=2cm,bottom=3cm,left=2.5cm,right=2.5cm,marginparwidth=1.75cm]{geometry}
\newtheorem{lemma}{Lemma}
\newcounter{claimctr}[lemma]
\newcounter{remctr}

\newenvironment{remark}{\medskip\par\noindent\refstepcounter{remctr}\textbf{Remark \theremctr.} }{\smallskip\par}
\newenvironment{myclaim}{\refstepcounter{claimctr}\medskip\par\noindent\textit{Claim \theclaimctr.}}{\par}

\newtheorem{defn}{{\bf Definition}}[section]

\newtheorem{theorem}[defn]{{\bf Theorem}} 
\newtheorem{corollary}{Corollary}

\title{New bounds on the anti-Ramsey numbers of star graphs}
\date{}
\author[1]{L. Sunil Chandran \footnote{part of this work was done
		while the author was visiting Max-Planck Institute for Informatik,
		Saarbruecken,\\ Germany, under Alexandr von Humboldt fellowship.}}
\author[1]{Talha Hashim}
\author[1]{Dalu Jacob}
\author[2]{Rogers Mathew}
\author[3]{Deepak Rajendraprasad}
\author[4]{Nitin Singh}
\affil[1]{Department of Computer Science and Automation\\ 
	Indian Institute of Science, Bangalore-560012, India \\
	\texttt{sunil@iisc.ac.in, talhahashim@iisc.ac.in, dalu1991@gmail.com}}
\affil[2]{Department of Computer Science and Engineering\\
	Indian Institute of Technology, Hyderabad-502285, India\\
	\texttt{rogers@cse.iith.ac.in}}
\affil[3]{Department of Computer Science and Engineering\\
	Indian Institute of Technology, Palakkad-678557, India\\
	\texttt{deepakmail@gmail.com}}
\affil[4]{IBM Research Lab, Manyata Embassy Business Park\\
	Bangalore-560045, India\\
	\texttt{nitisin1@in.ibm.com}}

\begin{document}
	
	\maketitle

	\begin{abstract}
		The  {\it anti-Ramsey number} $ar(G,H)$ with {\it input graph} $G$ and {\it pattern graph} $H$, is the maximum positive integer $k$ such that there exists an edge coloring of $G$ using
		$k$ colors, in which there are no {\it rainbow} subgraphs isomorphic to $H$ in $G$.  ($H$ is
		{\it rainbow} if all its edges get distinct colors). 
		The concept of  {\it anti-Ramsey number} 
		was introduced by Erd\"os, Simanovitz, and S\'os
		in 1973.  Thereafter several researchers investigated this concept in the
		{\it combinatorial} setting.  
		The cases where pattern graph $H$ is a complete graph $K_r$, a path $P_r$ or a star $K_{1,r}$ for  a fixed positive integer $r$,
		are well studied.
		Recently, Feng et al. revisited the  {\it anti-Ramsey}  problem for the pattern graph
		$K_{1,t}$ (for $t  \ge  3$) purely from an {\it algorithmic}
		point of view, due to its applications in {\it interference modeling} of
		{\it wireless networks}.  They posed it as an optimization problem, the \emph {maximum edge $q$-coloring problem.} 
		For a graph $G$ and an integer $q\geq 2$, an edge $q$-coloring of $G$ is an assignment of colors to edges of $G$, such that edges incident on a vertex span at most $q$ distinct colors. The {\it maximum edge} $q$-coloring problem seeks to {\it maximize} the number of colors in an edge $q$-coloring of the  graph $G$. 
		Note that the {\it optimum value} of the edge $q$-coloring problem of $G$  equals $ar(G,K_{1,q+1})$.  
		
		In this paper, we study $ar(G,K_{1,t})$, the {\it anti-Ramsey number} of stars, for each fixed integer $t\geq 3$,  both from  {\it combinatorial} and 
		{\it algorithmic} point of view.
		The first  of our main results presents  an upper bound for $ar(G,K_{1,q+1})$,
		in terms of number of vertices and the minimum degree of $G$. The second one improves this
		result for the case of triangle-free input graphs.
		
		For a positive integer $t$, let $H_t$ denote a subgraph of $G$ with maximum number of possible edges and maximum degree $t$.  From an observation of  Erd\"os, Simanovitz, and S\'os, we get:  $|E(H_{q-1})| + 1\leq ar(G,K_{1,q+1}) \leq |E(H_{q})|$. For instance, when $q=2$, the subgraph $E(H_{q-1})$ refers to a maximum matching.
		It looks like $|E(H_{q-1})|$ is the most natural parameter associated with the anti-Ramsey number $ar(G,K_{1,q+1})$, and the approximation algorithms
		for the maximum edge coloring problem usually proceed by first computing the $H_{q-1}$, then
		coloring all its edges with different colors and by giving one  (sometimes more than one)  extra colors to the remaining edges. The approximation guarantee of these algorithms usually depends on
		upper bounds for $ar(G,K_{1,q+1})$ in terms of $|E(H_{q-1})|$.
		
		Our third main result presents an upper bound for $ar(G,K_{1,q+1})$ in terms of  
		$|E(H_{q-1})|$. 
		
		All our results have algorithmic consequences. For some large special classes of graphs, 
		such as $d$-regular graphs, where $d\geq 4$, 
		our results can be used to prove a better approximation guarantee for the well-known approximation algorithm for the maximum edge coloring problem. We also show that all  our bounds are almost tight.
	\end{abstract}

	~~~~
	
	
	~~~
	
	\section{Introduction}

	The concept of {\it anti-Ramsey number} was introduced by Erd\"os and Simanovitz 
	in 1973~\cite {PEMS}.  
	A $k$-edge-coloring of $G$  is a function
	$c:E(G) \rightarrow [k]$. (Note that an edge coloring here  need not be proper.)
	A subgraph $G'$ of $G$ is called a {\it rainbow subgraph} with respect to an
	edge coloring $c$ of $G$, if all the edges of $G'$ are colored distinctly,
	i.e. no two edges of $G'$  share the same color in the coloring $c$. 
	The {\it anti-Ramsey number} $ar(G,H)$ is defined to be the largest integer $k$
	such that there exists an edge-coloring of $G$ using $k$ colors, in which there are no {\it rainbow} subgraphs isomorphic to $H$ in $G$. 
	In this paper, we will refer to the first parameter of $ar(G,H)$
	as {\it `input graph'} and 
	the second parameter as the {\it `pattern graph'}. Typically  the {\it pattern graph} $H$ is a {\it fixed graph} on $r$ vertices with simple structure,
	like the {\it complete graph} $K_r$, the path on $r$ vertices $P_r$ or the star
	on $r$ vertices, $K_{1, r}$. For convenience's sake, we  will call $ar(K_n,H)$, where the input graph is a complete graph as the \textit{classical anti-Ramsey number} of $H$, and $ar(G,H)$ where the input graph is {\it any} simple  graph on $n$ vertices as {\it the general anti-Ramsey number} of $H$.
	
	\medskip


	\medskip
	
	After Erd\"os and Simanovitz introduced {\it anti-Ramsey numbers}, several
	researchers investigated this concept. 
	A survey of known results is available in~\cite {FCO}. As mentioned before, the {\it pattern graphs} are usually graphs of very simple structure
	(like complete graphs, paths, stars, cycles) on a fixed number of vertices. 
	The case where the {\it pattern graph} is $P_{k}$, the path on $k$ vertices was studied
	by  Simanovitz and S\'oz~\cite {SS}, who gave an exact expression for the classical anti-Ramsey number of $P_k$, (i.e. $ar(K_n, P_k)$).  
	When the {\it pattern graph} is a cycle of length $k$ (denoted by $C_k)$, Erd\"os, Simanovitz and S\'os    
	proved that  $ar(K_n,C_3) = n-1$ and conjectured that  $ar(K_n, C_k) =
	\big(\frac {k-2}{2} + \frac {1} {k-1}\big) n + O(1) $ for $k \ge 4$~\cite {PEMS}.
	This was indeed verified by Alon et. al. \cite {Alon} for  {\it pattern graph}  $C_4$.   $ar(K_n, C_k)$ was also studied
	by Jiang and West \cite {JW}. 
	Montellano-Ballesteros and Neumann-Lara  verified the conjecture for any $k$~\cite {MBNL}, about 30 years after it was conjectured. 
	The case where the {\it pattern graph} is $K_r$, (i.e. $ar(K_n,K_r)$) was independently studied in~\cite {Schier} and~\cite {MBNL1}. Apart from the above, Manoussakis et. al.~\cite {MSTV} had studied $ar(K_n,K_{n-1})$.
	Another variant was considered in~\cite {MBNL2}: They studied rainbow stars within chosen
	subsets of vertices in colored multi-graphs.  
	The classical  anti-Ramsey number of $tP_2$ (i.e. $ar(K_n, tP_2)$, where $tP_2$ represent a matching consisting of $t$ edges)
	was studied by  Schiermeyer~\cite {Schier}, Fujita, Kaneko, Schiermeyer and Suzuki~\cite {FKSS}, Haas and Young~\cite {HY}.  
	The case where the {\it pattern graphs} are
	sub-divided graphs~\cite {Jiang2} was  studied by Jiang, and  trees of order $k$ 
	was studied by Jiang and West~\cite {JW1}.
	The case where the input graph is $K_{n,n}$ instead of $K_n$ was considered 
	with even cycles as {\it pattern graphs} by Axenovich et.al.~\cite {AJK}
	and stars as {\it pattern graphs} by Jiang~\cite {Jiang}. 
	
	\medskip
	
	The classical anti-Ramsey number of $K_{1,t}$ (i.e.  $ar(K_n,K_{1,t})$) and its
	bipartite analogue $ar(K_{n,n}, K_{1,t})$ is extensively studied in literature:
	Improving the bounds of~\cite {MSTV}, both Jiang~\cite {Jiang} and Montellano-Ballesteros~\cite {JJMB} proved the following
	formula for $ar(K_n,K_{1,t})$:

	\begin {eqnarray}
	\left \lfloor \frac {n(t-2)} {2} \right \rfloor  +  \left \lfloor \frac {n} {n-t+2} \right \rfloor   \le  ar   (K_n,  K_{1,t} ) \le \left \lfloor \frac {n(t-2)} {2} \right \rfloor  +  \left \lfloor \frac {n} {n-t+2} \right \rfloor  + 1 
	\end {eqnarray}

	The bipartite analogue was considered by Jiang~\cite {Jiang}, who proved the following:
	
	\begin {eqnarray}
	ar (K_{n,n}, K_{1,t} ) =  n (t-2) + \left \lfloor \frac {n} {n-t+2} \right \rfloor + 1   
	\end {eqnarray}

	In this paper, we study the {\it general anti-Ramsey number} of {\it star graphs}. i.e. we will consider only the star graphs as the {\it pattern graph},
	namely the class $K_{1,t}$ for any fixed integer $t\geq 3$. Though we restrict the {\it pattern graph},  we allow the
	{\it input graph} to be a {\it general graph} and study the {\it general anti-Ramsey number} of star graphs, $ar(G, K_{1,t})$. 
	It should be mentioned that Montellano-Ballesteros~\cite {JJMB} had studied the generalized
	{\it anti-Ramsey number} of stars, $ar(G, K_{1,t})$, and had proved an upper bound expression
	which depends on the value of a particular function defined on certain subgraphs of
	the given graph. Their expression is not easy to state or use, but they cleverly
	use that expression to derive the generalised {\it anti-Ramsey number} of stars for
	complete graphs, complete $r$-partite graphs, hypercubes, Cartesian products of
	two cycles, etc.
	The advantage of their approach is that it is an attempt to use a uniform  technique
	for the above-listed input graphs.  As far as we can see, it is not possible
	to get any simple formula as an upper bound 
	for $ar(G,K_{1,t})$, or even for $ar(G,K_{1,3})$ from their expression: rather it provides a method to
	try on special cases, wherever it works.

	\smallskip
	\noindent {\bf The edge $q$-coloring problem:}  Interestingly,
	the problem of finding $ar(G,K_{1,q+1})$, for integers $q\geq 2$, was revisited recently by  some applied researchers
	due to its applications in  modeling 
	channel assignment in networks equipped with multi-channel wireless interfaces~\cite{ARTC}. 
	For a graph $G$ and an integer $q\geq 2$, the edge $q$-coloring problem
	seeks to maximize the number of colors used to color the edges of $G$,
	subject to the constraint that a vertex $v\in V(G)$, is incident with edges of
	at most $q$ different colors.  It is easy to see that the number of colors
	in an optimal edge $q$-coloring  (hereafter refered to as {\rm OPT}) equals $ar(G,K_{1,q+1})$.  
	
	\smallskip
	
	The maximum edge $q$-coloring problem has been studied from
	an  algorithmic perspective. Adamaszek  and Popa~\cite{AAAP} proved that the problem is NP-hard for every
	$q\geq 2$. Moreover, they proved that assuming the {\it unique games
		conjecture}, it cannot be approximated within a factor less than $1+1/q$ for
	every $q\geq 2$ and assuming just $P\neq NP$, it cannot be approximated
	within a factor less than $1+(q-2)/(q-1)^2$ for every $q\geq 3$~\cite{PC}. 
	An algorithm (Algorithm~\ref{alg:feng}) with an approximation guarantee 2 for the maximum edge
	$2$-coloring for general graphs was given by Feng et. al. in~\cite{WFPC,WFLZ,WFLZ2}. The same algorithm has an approximation guarantee, $(1+\frac{4q-2}{3q^2-5q+2})$, for $q>2$.
	The algorithm
	of  Feng et. al.  is simple and intuitive and is given
	below (Algorithm~\ref{alg:feng}). 
	
	\medskip
	
	Note that for a positive integer $t$, a $t$-\emph{factor} of $G$ is a $t$-regular spanning subgraph of $G$. Now we define a $t$-\emph{sub-factor} of $G$, to be a sub-graph of $G$ with maximum degree $t$. Further, we call a subgraph of $G$ to be a \emph{maximum $t$-subfactor}, if it is a $t$-sub-factor of $G$ having the maximum number of possible edges. For instance, any maximum matching in $G$ is a maximum $1$-sub-factor of $G$, whereas a perfect matching (if exists) in $G$ is a 1-factor of $G$. Hereafter, we refer to Algorithm~\ref{alg:feng} as the {\em sub-factor based algorithm}. Recall that a maximum matching in a graph can be computed in polynomial time. As a generalization of this, Gabow~\cite{Gabow} proposed an $O(|V||E|\mathrm{log}|V|)$-time algorithm for finding a maximum $t$-sub-factor in an input graph. Consequently, Algorithm~\ref{alg:feng} runs in polynomial time.
	Feng et.al.~\cite{WFLZ2} also showed that the 2-edge coloring problem is
	polynomial-time solvable for trees and complete graphs. Adamaszek and Popa~\cite{AAAP}
	showed that the approximation guarantee of the sub-factor based algorithm (when $q=2)$ of Feng et. al. can be improved to
	$5/3$, for graphs with {\it perfect matching}. Later, for {\em triangle-free graphs with perfect matching}, this approximation guarantee (when $q=2$) was improved to $8/5$ by Chandran et.al.~\cite{Chandran}. 
	
	\begin{algorithm}
		\label{alg:feng}
	\end{algorithm}
	\begin{tcolorbox}
		\textbf{\underline {Algorithm~1}}\vspace{.4cm}\\
		\textbf{Input:} A graph $\mathbf{G}$ and a positive integer $\mathbf{q\geq 2}$
		
		\medskip
		Find a maximum $(q-1)$-sub-factor, say $M$ of $G$;
		
		\smallskip
		To each edge in $M$ assign a new color;
		
		\smallskip
		For each non-trivial connected component $C$ of the subgraph $G'=G\setminus M$ of $G$, \\
		assign a new single color to all the edges in $C$; 
		
		\medskip
		\textbf{Output:} each edge along with the color assigned to it.
	\end{tcolorbox}


	\subsection {Our Results:} \label{results}
	
	Our contributions in this  paper are both 
	combinatorial and algorithmic: On the one hand, we give bounds for
	the general anti-Ramsey number of stars in the same style as~\cite {Jiang,MSTV,JJMB} and, on the other hand,
	we improve  the approximation guarantee
	of the sub-factor based algorithm (Algorithm~\ref{alg:feng}) 
	significantly for graphs with a large minimum degree, supplementing the
	works of~\cite {AAAP,WFPC,WFLZ,WFLZ2}. 
	
	Specifically, we obtain
	the following results: Results~1 and  2
	supplements the traditional literature on anti-Ramsey numbers,
	whereas Result~3, though essentially a combinatorial bound, is motivated by
	improving the approximation guarantee for large special classes of graphs, such as regular graphs and 
	graphs of minimum degree greater than $C \sqrt n$. Result~$3^\prime$ is a generalization of Result~3 for each integer $q\geq2$. We also propose two constructions for families of graphs to demonstrate that the bounds we obtained are \textit{almost} tight. First, we note the following comment.
	
	\medskip
	
	\noindent {\bf  Comment:} 
	Given a family  of graphs ${\cal H}$, let the Turan-type extremal number 
	$ex(G, {\cal H} )$ be  defined as  the maximum number
	of edges in a subgraph of $G$  containing no copy of $H$ for any 
	$H \in {\cal H}$. When ${\cal H} = \{ H \}$, we may just write $ex(G,H)$
	instead of $ex(G, {\cal H})$.  
	Erd\"os et al.~\cite{PEMS} related the above extremal numbers of  Turan  with
	general anti-Ramsey numbers by proving the following~\cite{Axenovich}: 
	
	\begin {eqnarray}
	\label {Turan-anti-Ramsey-relation}
	\mbox {For input graph $G$ and pattern graph $H$:}  \nonumber\\
	ex(G, {\cal H}) + 1 \le ar(G,H) \le ex(G,H), \label{eqn:extremal}\\
	\mbox {where } {\cal H} = \{ H - \{e\} : e \in E(H)  \} \nonumber
	\end {eqnarray}
	
	Note that when $H=K_{1,q+1}$, for some integer $q\geq 2$, $ex(G,H)$ refers to a maximum $q$-sub-factor in $G$, and $ex(G,{\cal H})$ refers
	to maximum $(q-1)$-sub-factor in $G$. For a positive integer $t$, let $H_t$ denote a maximum $t$-sub-factor in $G$. Therefore, when $H=K_{1,q+1}$, by the above remark we have $ex(G,H)=|E(H_q)|$ and $ex(G,{\cal H})=|E(H_{q-1})|$. For any integer $q\geq 2$, we then have the following by \eqref{eqn:extremal}: 
	
	\begin {equation}
	|E(H_{q-1})| + 1\leq ar(G,K_{1,q+1}) \leq |E(H_{q})| \label{eqn:maxsubgraph}
	\end {equation}
	~~~~~~
	Note that inequality~\eqref{eqn:maxsubgraph}
	indicates that
	$|E(H_{q-1})|$, the cardinality of a maximum $(q-1)$-sub-factor in $G$ is the most natural parameter
	associated with $ar(G,K_{1,q+1})$.  It is usual in graph theory to  attempt to tighten the
	bounds in terms  of such naturally associated parameters- for example, 
	the Vizing's  theorem was a result of attempts to  tighten the upper bound for the chromatic index in terms of
	a naturally associated lower bound, namely the maximum degree.  Also,
	since a maximum matching, and in general, a maximum $(q-1)$-sub-factor in graphs, can be efficiently computed~\cite{Gabow}, the inequality~\eqref{eqn:maxsubgraph}
	inspires to look for approximation algorithms. The
	sub-factor based algorithm proposed by Feng et. al.~\cite {WFPC,WFLZ,WFLZ2} is one such algorithm. 
	
	\medskip
	
	In the remaining part of the section, we discuss our major contributions in the paper.
	
	\medskip
	
	\noindent {\bf Result 1:} Let $G$ be a graph with minimum degree $\delta \geq q+1$. We then have the following by  Theorem~\ref{thm:factor}(\ref{factor}): 
	\begin{equation*}
		ar(G,K_{1,q+1}) \le \frac{n(q-1)}{2}\bigg(1+\frac{2}{(q-1)(\delta-q+2)}\bigg)
	\end{equation*}
	
	In addition, if $G$ has a $(q-1)$-factor, by \eqref{eqn:maxsubgraph} we then have that :
	
	\begin{equation} \label{eqn:li}
		\frac{n(q-1)}{2}+1\leq ar(G,K_{1,q+1}) \le \frac{n(q-1)}{2}\bigg(1+\frac{2}{(q-1)(\delta-q+2)}\bigg)
	\end{equation}
	
	\medskip
	
	Note that, it follows from a result by Li and Cheng~\cite{Li} that every graph $G$ with number of vertices, $n\geq 4(q-1)-1$ and minimum degree $\delta\geq n/2$ has a $(q-1)$-factor. For such a graph $G$, we can then derive the following by \eqref{eqn:li}.
	\begin{equation}
		\frac{n(q-1)}{2}+1 \leq ar(G,K_{1,q+1}) \le 4+  \frac{n(q-1)}{2} - \frac{4}{q}
	\end{equation}
	
	In particular, for $q=2$ case, we note the following additional observations.
	
	\smallskip
	
	Let $G$ be any graph. If $F$ is a $2$-sub-factor of $G$ (i.e., a collection of paths and
	cycles),  there exists a matching, $M_F$ in $F$ of cardinality at least $|F|/3$, 
	and therefore $|E(F)| \le 3 |M_F|$. By \eqref{eqn:maxsubgraph} we then have the following:
	
	\begin {eqnarray}
	\label {lower-upper-bound-matching}
	|M| + 1 \le    ar(G,K_{1,3}) \le 3 |M|  \label{eqn:matching}\\
	\mbox {where M is a maximum matching of G.} \nonumber 
	\end {eqnarray}
	
	The lower bound of $|M| + 1$ along with Result~1 (for $q=2$), 
	allows us to derive the exact value
	of $ar(G,K_{1,3})$ for graphs of minimum degree  $> n/2$ when $n$ is
	even.
	This is because in graphs of minimum degree  $\delta >n/2$,  there exists a hamiltonian cycle 
	(by Dirac's theorem~\cite {Dirac}) and 
	therefore, a perfect matching if $n$ is even.  Therefore
	we get $n/2 + 1\le ar (G,K_{1,3})  \le n/2 + n/\delta <  n/2 + 2$. Since
	$ar(G,K_{1,3})$ is an integer we infer that $ar(G,K_{1,3}) = n/2+1$. When
	$n$ is odd, we can infer that $ar(G, K_{1,3}) \in \{  \left \lfloor n/2 \right
	\rfloor +  1, \left \lfloor n/2 \right
	\rfloor +  2\}$.

	~~~~~~
	
	\noindent {\bf Result 2:} For a triangle-free graph $G$ with $\delta \geq q+1$  
	we show a better upper bound by  Theorem~\ref{thm:factor}(\ref{triangle-free}):
	
	\begin{equation*}
		ar(G,K_{1,q+1}) \le \frac{n(q-1)}{2}\bigg(1+\frac{1}{(q-1)(\delta-q+1)}\bigg)
	\end{equation*}
	
	Again as in Result~1, in addition, if $G$ has a $(q-1)$-factor, by \eqref{eqn:maxsubgraph} we have that :
	
	\begin{equation*} 
		\frac{n(q-1)}{2}+1\leq ar(G,K_{1,q+1}) \le \frac{n(q-1)}{2}\bigg(1+\frac{1}{(q-1)(\delta-q+2)}\bigg)
	\end{equation*}
	
	~~~~~~~~~

	\noindent {\bf Implication of Result~1 and Result~2  to the approximation
		guarantee of the sub-factor based Algorithm (when $q=2$):} Result~1 implies that, when $q=2$,
	for graphs with a perfect matching, the sub-factor
	based algorithm (Algorithm~\ref{alg:feng}) has an approximation guarantee, $(1 + \frac {2} {\delta})$
	where $\delta$ is the minimum degree.  Recall that Adamaszek and Popa~\cite {AAAP} proved
	an approximation guarantee, $\frac {5}{3}$ for this case.  The approximation
	guarantee proved by us is better than that of Adamaszek and Popa  for 
	$\delta \ge 4$, and equals to theirs for $\delta = 3$.  For triangle-free
	graphs, our approximation guarantee is even better: $(1 + \frac {1} {\delta -1})$.
	This is better than the $5/3$-guarantee of~\cite {AAAP} and the $8/5$-guarantee of~\cite{Chandran} even for $\delta \geq 3$ .
	Note that the analysis of~\cite {AAAP} cannot be specialized for
	the triangle-free case to achieve a better approximation guarantee, as far as
	we understand. Further, as mentioned earlier,  when the minimum degree $> n/2$,
	Algorithm~\ref{alg:feng} is
	almost optimal, in fact,  an additive 1-approximation algorithm for odd $n$,
	and optimal for even $n$.

	~~~~~

	Observe that a 3-approximation guarantee is obvious from~\eqref{lower-upper-bound-matching}. 
	Feng et. al~\cite {WFPC,WFLZ,WFLZ2} proved that the approximation guarantee of the sub-factor based algorithm is
	2, by improving the upper bound of~\eqref{eqn:matching} 
	from $3|M|$ to $2|M|$. 
	Thus we have the following in the place of inequality~\eqref{eqn:matching}:
	
	\begin {eqnarray}
	\label {improved-bounds-matching}
	|M| + 1 \le ar(G,K_{1,3}) \le 2 |M|
	\end {eqnarray}

	Unfortunately, in general, this upper
	bound of $2|M|$ cannot be improved- for example, consider an even cycle:
	It has a perfect matching, whereas $ar(C_n,K_{1,3}) = n$. 
	\footnote {
		If you are curious how a $5/3$-approximation guarantee was proved
		in~\cite{AAAP}, please note that it is not by proving $ar(G,K_{1,3})
		\le  (5/3) |M|$ when $M$ is a perfect matching. They cleverly make use
		of another aspect of the sub-factor based algorithm, namely the number of
		components produced when certain perfect matching is removed. } 
	Looking at examples like cycles, at first, we thought that if we assume a high minimum degree,  we may get a better
	upper bound for $ar(G,K_{1,3})$. That is, a bound in terms of $|M|$ and 
	$\delta$ in the same spirit as that of Result~1.  But soon  we realized that
	if we fix $\delta$, it is possible to  get a graph $G$, such that 
	$ar(G,K_{1,3})$ is very close to  $2|M|$, but just that the cardinality of  its maximum
	matching is correspondingly small  with respect to the number of vertices $n$.
	(See the construction in Section~\ref{sec:tightkappa} and take $\kappa = (\delta +1)/2$.) 
	Thus we realized that we also need to consider another parameter: $\kappa = \frac {n} {2|M|}$. Note that for graphs with a
	perfect matching, $\kappa = 1$, and it increases as the fraction of matched
	vertices decreases. Thus, we also intend to find an upper bound  for $ar(G,K_{1,3})$ in terms
	of $|M|$, which also involves both $\kappa$ and the minimum degree $\delta$.  
	
	~~~~~~~~~~~
	
	\noindent {\bf Result 3:}  For a graph $G$ with a maximum matching $M$,
	$ ar(G,K_{1,3})\le |M| (1+ \frac {2(\kappa+1)} {\delta-1})$ (Theorem~\ref{thm:main}). 
	
	\smallskip
	It is evident that our upper bound is better than $2|M|$ as long as the cardinality of the
	maximum matching is not too small: i.e., as long as
	$\kappa <  (\delta - 3)/2$.  We may also wonder about the cases in which Result~3 has an upper hand over Result~1. It is worth noting that Result~1 (when $q=2$), is not useful
	to get any useful upper bound in terms of $|M|$ when $\kappa$ is even $2$:
	The upper bound implied in terms of maximum matching would be $> 2 |M|$.
	
	~~~~

	\noindent {\bf  Implication of Result 3  to the approximation
		guarantee of the sub-factor based Algorithm (when $q=2$):}
	From Result~3, for the general graphs, we get that
	the sub-factor based algorithm (when $q=2$) has an approximation guarantee, $(1+ \frac {2(\kappa+1)} {\delta-1})$. Here  $\kappa = n/2|M|$, where
	$n$ is the number of vertices in $G$ and $M$ is a maximum matching in $G$.
	Considering the attempts of Adamaszek and Popa to get a better approximation
	guarantee for graphs with a perfect matching, it is natural to consider this
	ratio. The approximation guarantee that we obtained using the parameter $\kappa$ is better than the previously known 2-approximation guarantee for the general graphs (for $q=2$)
	when $\kappa < (\delta - 3)/2$.
	
	\medskip
	
	From the technical looking $(1 +2(\kappa+1)/(\delta-1))$-approximation guarantee,
	we can easily get the following corollaries for some interesting special cases:
	
	\begin {itemize}

	\item  Corollary 1 of Result~3: Let $M$ be a maximum matching in a $d$-regular graph on $n$ vertices. We then have, $|M|\geq \frac{nd/2}{d+1}$. Therefore, $\kappa \leq 1+1/d$. Thus for  $d$-regular graphs, we have the approximation guarantee to be 
	$\big(1+\frac{2(2+1/d)}{d-1}\big)=\big(1 + \frac{2(2d+1)}{d(d-1)}\big)$. 
	This is better than the 2-approximation guarantee known for the
	general case~\cite{WFPC,WFLZ,WFLZ2}, when $d \ge 6$. (Note that all $d$-regular graphs
	need not have a perfect matching and thus 5/3 factor of~\cite {AAAP}  or the $(1 + 2/\delta)=(1+2/d)$-factor
	from our Result~1, is not
	applicable.)   
	
	\item Corollary 2 of Result~3: It is easy to see that $\kappa \le \frac {(\Delta+ \delta) } {\delta}$, where $\Delta$ is the
	maximum degree of $G$. For graphs of minimum degree $ h \sqrt n + 1$, we have an approximation
	guarantee, $(1 + \frac {6} {h^2})$.  To see this note that $ \kappa \le  \bigg(\frac {\Delta + \delta} {\delta}\bigg) \le \bigg(\frac {\Delta} {\delta} + 1\bigg)$.
	Now substituting  $\Delta \le n$ and $\delta \ge h \sqrt n + 1$, we see that the approximation
	guarantee is at most  $(1 + \frac {6} {h^2})$.  For $h > \sqrt 6 $, this is better
	than the 2-approximation guarantee known for the general graphs~\cite{WFPC,WFLZ,WFLZ2}.
	\end {itemize}
	
	Note that, for $q\geq 2$, a maximum $(q-1)$-sub-factor in $G$ is a natural generalization of a maximum matching in $G$. Also, as in Result~1 and 2, for $q\geq 2$, analogous to a perfect matching in $G$, we can consider a $(q-1)$-factor in $G$. In the following result (which is a generalization of Result~3), we set the parameter, $\kappa= \frac{n(q-1)}{2|E(H_{q-1})|}$, where $|E(H_{q-1})|$ is the number of edges in a maximum $(q-1)$-sub-factor of the given graph $G$. Note that, as any $(q-1)$-factor in a graph on $n$ vertices is of size $\frac{n(q-1)}{2}$, we can observe that for graphs with a $(q-1)$-factor, we have
	$\kappa = 1$, and it increases as the fraction of the vertices belonging to the $(q-1)$-factor decreases. 
	
	\medskip
	\noindent {\bf Result~$3^\prime$:} For general graphs
	we have the following upper bound: (by Theorem~\ref{thm:main} and~\ref{thm:main2}) 
	\begin{equation*}
		ar(G,K_{1,q+1}) \le|E({H_{q-1})}|\bigg( 1+ \frac {2\big(\kappa+\frac{1}{(q-1)}\big)} {\delta-1}\bigg), 
	\end{equation*}
	where $|E(H_{q-1})|$ is the number of edges in a maximum $(q-1)$-sub-factor of $G$ and $\kappa= \frac{n(q-1)}{2|E(H_{q-1})|}$. 
	
	\medskip
	
	We then have the following implication for a maximum edge $q$-coloring problem for each integer $q>2$. i.e. when $\kappa < \frac{1}{(q-1)}\bigg(\frac{(2q-1)(\delta-1)}{(3q-2)}-1\bigg)$, the bound provided in this result improves the previously known approximation guarantee, $\bigg(1+\frac{4q-2}{3q^2-5q+2}\bigg)$ given by Feng et.al.

	~~~~~~~~~

	\section{Notation and Preliminaries}
	
	Throughout this paper, we consider connected graphs with minimum degree
	$\delta\geq q+1$, where $q\geq 2$. An edge $q$-coloring of a graph $G$ with $c$ colors is a
	map ${\cal C}: E(G)\rightarrow [c]$, such that a vertex is incident with
	edges of at most $q$ colors. Let ${\rm ALG}(G)$ denote the number of colors in the
	coloring returned by
	Algorithm~\ref{alg:feng} for the graph $G$, and let ${\rm OPT}(G)$ be the
	number of colors in a maximum edge $q$-coloring of $G$. 
	
	\subsection{Characteristic subgraph for edge $q$-coloring}
	
	Let ${\cal C}$ be an edge $q$-coloring of the graph $G$. We are interested in subgraphs of $G$
	containing exactly one edge of each color in ${\cal C}$. Note that the maximum degree of such a subgraph is at
	most $q$.
	
	\smallskip
	
	\noindent {\bf Characteristic Subgraph:}  
	We define the {\em characteristic subgraph} of $G$ with respect to a coloring ${\cal C}$
	as a subgraph of $G$ containing exactly one edge of each color in ${\cal C}$. Let ${\cal C}$ be an optimal edge $q$-coloring of $G$, and let $\chi$ be a
	characteristic subgraph of $G$ with respect to the coloring ${\cal C}$. For $i\in \{0,1,2,
	\ldots,q\}$, let
	$N_i$ denote the vertices of $G$ with degree $i$ in
	$\chi$ and let $n_i=|N_i|$. Clearly $n_0+n_1+n_2+\ldots+n_q=n$.
	
	\section{Graphs with a $(q-1)$-factor}
	In this section, we first derive an upper bound for $ar(G,K_{1,q+1})$, where $q\geq 2$. As a consequence, we obtain some useful bounds for the approximation guarantee of the sub-factor based algorithm, when it is applied to graphs with a $(q-1)$-factor. The proofs here also help to illustrate key ideas in a simpler setting, which will be extended further in the upcoming sections. 
	
	\begin{theorem}\label{thm:factor}
		For a fixed integer $q\geq 2$, let $G$ be an $n$-vertex graph with a $(q-1)$-factor and have minimum degree $\delta(G)=\delta \ge q+1$. We then have:
		\begin{enumerate}[{\rm (i)}]
			\item \label{factor}   $ar(G,K_{1,q+1}) \le \frac{n(q-1)}{2}\bigg(1+\frac{2}{(q-1)(\delta-q+2)}\bigg)$
			
			\item \label{triangle-free}  Further, if $G$ is also triangle-free, we have $ar(G,K_{1,q+1}) \le \frac{n(q-1)}{2}\bigg(1+\frac{1}{(q-1)(\delta-q+1)}\bigg)$.
			
		\end{enumerate}
	\end{theorem}
	\begin{proof}[Proof of (\ref{factor})]
		Let ${\cal C}$ be an optimal edge $q$-coloring of $G$ using, say $c$ colors. We then have, $c=ar(G,K_{1,q+1})$. Let $\chi$ be a
		characteristic subgraph of $G$ with respect to the coloring ${\cal C}$. Recall that for
		$i\in \{0,1,2,
		\ldots,q\}$, $N_i$ denote the vertices of $G$ with degree $i$ in
		$\chi$ and $n_i=|N_i|$. We also have $n_0+n_1+n_2+\ldots+n_q=n$. 
		
		\medskip
		
		Note that the number of colors, $c$ used in the coloring ${\cal C}$
		is the same as the number of edges in $\chi$ and is given by:
		\begin{align}
			c & = \frac{qn_{q}+(q-1)n_{q-1}+\cdots+2n_{2}+n_{1}}{2} \nonumber\\
			& = \frac{(q-1)(n_{q}+n_{q-1}+\cdots+n_{2}+n_{1}+n_{0})+ (n_{q}-n_{q-2}-2n_{q-3}-3n_{q-4}-\cdots-(q-1)n_{0})}{2}\nonumber\\ 
			& = \frac{n(q-1)}{2}+ \frac{n_{q}-n_{q-2}-2n_{q-3}-3n_{q-4}-\cdots-(q-1)n_{0}}{2}  \label{eqn:c2}
		\end{align}
		For $v\in N_q$, let $N'(v)$ denote the neighbors of $v$ in $G$ through
		edges not in $\chi$. Clearly $|N'(v)|\geq \delta-q$ for each $v\in N_q$.
		
		\medskip

		\begin{myclaim}
			For $u,v\in N_q$, $v \notin N'(u)$. 
		\end{myclaim}
		
		\smallskip
		\noindent {\it Proof of Claim:} Let ${\cal
			C}_u$ and ${\cal C}_v$ be the set of colors incident at vertices $u$ and $v$
		respectively. If $v  \in N'(u)$ then  the color of $uv$ belongs to
		${\cal C}_u\cap {\cal C}_v$.  Since $u,v\in N_q$, by the definition of the characteristic graph, we have that
		if $uv$ is not an edge in $\chi$,  then $C_u \cap C_v = \emptyset$. 
		This implies that the edge $uv$ cannot get a color if $v \in N'(u)$. Thus we infer
		that $v \notin N'(u)$.
		\medskip
		
		By the above claim, it follows that for each $v\in
		N_q$, we have $N'(v)\subseteq N_0\cup N_1\cup\ldots\cup N_{q-1}$. Consider the bipartite graph $H$ with bipartition $N_q\uplus (N_0\cup
		N_1\cup\ldots\cup N_{q-1})$ and edge set of $H$, $E(H)$ given by $E(H) := \cup_{v\in N_q}
		E(v,N'(v))$.
		Let $i\in \{0,1,2,\ldots,q-1\}$. We now show that $d_H(u)\leq 2(q-i)$ for each vertex $u\in N_i$. Consider a vertex $u\in N_i$. Now, by the definition of $N_i$ and $H$, there are at most $q-i$ colors, say $a_1,a_2,
		\ldots, a_{q-i}$ in ${\cal C}_u$,
		but not incident to $u$ in $\chi$. Let $w$ be a neighbor of $u$ in $H$. Then 
		we must have that $a_j\in {\cal C}_w$ for some $j\in \{1,2,\ldots,q-i\}$. Since there are at
		most two vertices, say $v,w\in N_q$ such that $a_j\in {\cal C}_v\cap {\cal C}_w$, it follows that
		$u$ has at most $2(q-i)$ neighbors in $H$. i.e. $d_H(u)\leq 2(q-i)$ for each $u\in N_i$.
		
		\medskip
		Now by counting the edges across the bipartition, $N_{q}\uplus (N_{0}\cup N_{1}\cup N_{2}\cup \ldots \cup N_{(q-1)})$ in two ways, we have:\\
		\begin{align}
			n_{q}(\delta - q) & \leq 2qn_{0}+2(q-1)n_{1}+2(q-2)n_{2}+\cdots+2n_{q-1} \label{eqn:degree}\\
			n_{q}(\delta - q+2) & \leq (2n_{q}+ 2n_{q-1}+2n_{q-2}+\cdots+2n_{0})+2n_{q-2}+4n_{q-3}+\cdots+2(q-1)n_{0} \nonumber\\
			& \leq 2n+ 2(n_{q-2}+2n_{q-3}+\cdots+(q-1)n_{0}) \nonumber\\
			\frac{n_{q}(\delta-q+2)}{2} & \leq n+ (n_{q-2}+2n_{q-3}+\cdots+(q-1)n_{0}) \label{eqn:n_q}
		\end{align}
		Note that,
		\begin{align*}
			\frac{(\delta-q+2)(n_{q}-n_{q-2}-2n_{q-3}-3n_{q-4}-\cdots -(q-1)n_{0})}{2} &\leq n_{q}\bigg(\frac{\delta-q+2}{2}\bigg)-(n_{q-2}+2n_{q-3}+3n_{q-4}+ \\&\cdots+(q-1)n_{0})\\ &\leq n \hspace{.5cm} \text{ (since }\delta\geq q+1  \text{ and by equation }~(\ref{eqn:n_q}))
		\end{align*}
		Therefore, we have,
		\begin{align}
			\frac{n_{q}-n_{q-2}-2n_{q-3}-3n_{q-4}-\cdots-(q-1)n_{0}}{2} & \leq \frac{n}{(\delta-q+2)}  \label{eqn:c3}
		\end{align}
		Now by equations (\ref{eqn:c2}) and (\ref{eqn:c3}), we have:
		\begin{align*}
			c  & \leq \frac{n(q-1)}{2} + \frac{n}{\delta-q+2}\\
			& \leq \frac{n(q-1)}{2}\bigg(1 + \frac{2}{(q-1)(\delta-q+2)}\bigg)\\
		\end{align*}
		Since $c=ar(G,K_{1,q+1})$, we then have $ar(G,K_{1,q+1})\le \frac{n(q-1)}{2}\bigg(1+\frac{2}{(q-1)(\delta-q+2)}\bigg)$. This proves Part~(i) of the theorem.

		\medskip
		
		\noindent \textit{Proof of (\ref{triangle-free})}. Let $i\in \{0,1,2,\ldots,q-1\}$ and $u\in N_i$. As in Part~(i), by the definition of $N_i$ and $H$, we know that there are at most $q-i$ colors available at $u$ to color the edges incident to $u$ in $H$. Since here, the graph $G$ is further assumed to be triangle-free, we
		now show that $d_H(u)\leq (q-i)$. This can be achieved by proving that the vertex $u$ can have at most one edge of a given
		color incident on it in $H$. Let $a$ be one of
		the colors incident to $u$ in ${\cal C}$. For the sake of contradiction, assume that there exist distinct vertices $x,y\in N_q$ such that $ux,uy\in E(H)$ and both the edges $ux$ and $uy$ are colored $a$. Since $x,y\in N_q$ and $a\in {\cal C}_x\cap
		{\cal C}_y$, we then have that $xy$ is an edge in $\chi$. This is a contradiction to the fact that $G$ is triangle-free.
		Therefore, we can conclude that for each $u\in N_i$, where $i\in \{0,1,2,\ldots,q-1\}$, $d_H(u)\leq q-i$. Thus, instead of equation~(\ref{eqn:degree}) in Part~(i), we now have the following:
		\begin{align}
			n_{q}(\delta - q) & \leq qn_{0}+(q-1)n_{1}+(q-2)n_{2}+\cdots+n_{q-1} \nonumber\\\
			n_{q}(\delta - q+1) & \leq (n_{q}+ n_{q-1}+n_{q-2}+\cdots+n_{0})+n_{q-2}+2n_{q-3}+\cdots+(q-1)n_{0}\nonumber\\
			n_{q}(\delta-q+1) & \leq n+ (n_{q-2}+2n_{q-3}+\cdots+(q-1)n_{0}) \label{eqn:n_q_triangle}
		\end{align}
		
		Making similar substitutions as before, we now have,
		\begin{align*}
			(\delta-q+1)(n_{q}-n_{q-2}-2n_{q-3}-3n_{q-4}-\cdots-(q-1)n_{0}) &\leq n_{q}(\delta-q+1)-(n_{q-2}+2n_{q-3}+3n_{q-4}+\\ &\cdots+(q-1)n_{0})\\ &\leq n \hspace{.5cm} \text{ (since }\delta\geq q+1  \text{ and by equation}~(\ref{eqn:n_q_triangle}))
		\end{align*}
		
		Therefore, we have,
		\begin{align}
			n_{q}-n_{q-2}-2n_{q-3}-3n_{q-4}-\cdots-(q-1)n_{0} & \leq \frac{n}{(\delta-q+1)} \label{eqn:c3triangle}
		\end{align}
		Now by equations (\ref{eqn:c2}) and (\ref{eqn:c3triangle}), we have:
		\begin{align*}
			c  & \leq \frac{n(q-1)}{2} + \frac{n}{2(\delta-q+1)}\\
			& \leq \frac{n(q-1)}{2}\bigg(1 + \frac{1}{(q-1)(\delta-q+1)}\bigg)\\
		\end{align*}
		Thus, we can conclude that, $ar(G,K_{1,q+1})\leq \frac{n(q-1)}{2}\bigg(1 + \frac{1}{(q-1)(\delta-q+1)}\bigg)$. This proves Part~(ii) of the theorem.
	\end{proof}
	\begin{corollary}\label{corr:factor}
		For a fixed integer $q\geq 2$, let $G$ be an $n$-vertex graph with a $(q-1)$-factor and with minimum degree,
		$\delta(G) = \delta \ge q+1$. We then have:
		\begin{enumerate}[{\rm (i)}]
			\item  ${\rm OPT}(G)\leq \bigg(1 + \frac{2}{(q-1)(\delta-q+2)}\bigg)\cdot{\rm ALG}(G)$.
			
			\item  If $G$ is triangle-free,  then ${\rm OPT}(G)\leq  \bigg(1 + \frac{1}{(q-1)(\delta-q+1)}\bigg)\cdot{\rm ALG}(G)$.
			
		\end{enumerate}
	\end{corollary}
	\begin{proof}
		Since $G$ has a $(q-1)$-factor, we have ${\rm ALG}(G)\geq \frac{n(q-1)}{2}$. Now the corollary is immediate from Theorem~\ref{thm:factor} and the fact that ${\rm OPT}(G)=ar(G,K_{1,q+1})$.
	\end{proof}
	
	\medskip 
	\begin{corollary}
		For the case $q=2$, ${\rm OPT}(G)\leq {\rm ALG}(G)+1$ when $\delta\geq \lfloor n/2
		\rfloor$. Moreover, if $n$ is even and $\delta > n/2$, then Algorithm~\ref{alg:feng} (for $q=2$) is optimal, i.e.  ${\rm ALG}(G) = {\rm OPT} (G)$.  
	\end{corollary}
	\begin{proof}
		Observe that by Dirac's theorem, a graph $G$ with
		$\delta \ge  \lfloor n/2 \rfloor$ has a hamiltonian cycle, and hence a maximum
		matching of size $\lfloor n/2 \rfloor$. Thus ${\rm ALG}(G)\geq \lfloor n/2
		\rfloor + 1$.  Further from Part~(i),  we have $c\leq n/2 + n/\delta <
		\lfloor n/2 \rfloor +
		3$. Thus, since $c$ is an integer we get 
		$c\leq \lfloor n/2 \rfloor + 2\leq {\rm ALG}(G)+1$. Note that if we assume
		that $n$ is even and $\delta > n/2$, this proves that Algorithm~\ref{alg:feng}
		is optimal in this case.
		
	\end{proof}
	\subsection{Tight example for Theorem~\ref{thm:factor}}
	In this section, we propose a construction of a family of graphs for which the sub-factor based algorithm can end up with  an approximation guarantee almost equal to the expressions obtained in Theorem~\ref{thm:factor} and Corrollary~\ref{corr:factor}. This shows that the expressions for the approximation guarantee that we provided in Theorem~\ref{thm:factor} and Corrollary~\ref{corr:factor} are almost tight.
	To describe the construction, we need the following concept.
	\medskip
	
	\noindent{\bf The degree-preserving lift  $L(G)$ of a $d$-regular graph $G$:}
	Let $G$ be a $d$-regular graph. For each vertex $v\in V(G)$ there are exactly $d$ edges, say $e_v^1,e_v^2,\ldots, e_v^d$ incident at $v$. Now, for each vertex $v\in V(G)$, we introduce $d$ vertices, say $v_1,v_2,\ldots,v_d$, where the vertex $v_i$ ``represents'' the edge $e_v^i$ incident at $v$. We then construct the graph $L(G)$ from $G$ by replacing each vertex $v$ in $G$ by a clique, say $K_v$ on these $d$ vertices, $v_1,v_2,\ldots,v_d$. Clearly, in $L(G)$, all the vertices $v_i$ has $d-1$ neighbours in its clique  $K_v$. Further, for $u\neq v$, two vertices $u_i\in K_u$ and $v_j\in K_v$ are connected by an edge $u_iv_j$ in $L(G)$, if and only if $uv \in E(G)$ and  both $u_i$ and $v_i$ ``represents'' the same edge, i.e. $e_u^i=e_v^j= uv$.  Thus apart from the $d-1$ neighbours
	from inside its clique $K_v$, the vertex $v_i$ gets exactly one more neighbour from outside $K_v$.
	\emph {Thus it is easy to verify that $L(G)$ is again a $d$-regular graph, whose $n'=nd$ vertices are partitioned into $n= \frac{n'} {d}$ cliques,  each on $d$ vertices,  such that each vertex in any of these cliques, say $K$  has exactly one neighbor in $V(L(G))\setminus K$. }

	\medskip
	
	Given a fixed integer $q\geq 2$, let $d\geq 2q+1$ be any integer such that $d-q+2$ is even. Let $G$ be any $d$-regular graph on $n$ vertices. Let $G_0=G$. For each $i\in \{1,2,\ldots,q-1\}$, we will recursively construct a graph $G_i$ with $|V(G_i)|=n_i$ satisfying the following property: 
	\begin{multline} \label{graphproperty}
		V(G_i) \mbox{ can be partitioned into } l_i=\frac{n_i}{d-i+1} \mbox{ cliques, say } K_1,K_2,\ldots, K_{l_{i}},  \mbox{ where each clique } K_j,\\  j\in \{1,2,\ldots,l_i\} \mbox{ is of size } d-i+1, \mbox{  and each vertex } v\in K_j \mbox{ has } i \mbox{ neighbors in } V(G_i)\setminus K_j.  
	\end{multline}
	
	Note that the graphs $G_i$ satisfying the Property~\eqref{graphproperty}  is a $d$-regular graph, since any vertex $v\in K_j$ has $d-i$ neighbors in $K_j$ and $i$ neighbors in $V(G_i)\setminus K_j$. We call  the edges of $G_i$ inside a clique  $K_j, 1 \le j \le l_i $  as  \emph {clique-edges}
	and the remaining edges (having end points in two different cliques) as \emph {cross edges.} Thus at each vertex of $G_i$, exactly $(d-i)$ clique edges and $i$ cross edges are incident.
	
	\medskip
	
	Now, for each $i\in \{1,2,\ldots,q-1\}$, we define the graph $G_i$ as follows:
	
	\smallskip
	
	Define $G_1 = L(G_0)$, the degree-preserving lift of $G_0$, where $G_0=G$.  
	It is clear that due to the properties of degree preserving lift,  $G_1$ satisfies Property~\eqref{graphproperty}.
	
	\smallskip
	
	For $i>1$, we recursively define the graph $G_i$ as follows: We can assume that $V(G_{i-1})$ satisfies Property~\eqref{graphproperty}. i.e. $V(G_{i-1})$ can be partitioned into $l_{i-1}=\frac{n_{i-1}}{d-i+2}$ cliques, say $K_1,K_2,\ldots, K_{l_{i-1}}$. For any $j\in \{1,2,\ldots,l_{i-1}\}$, we also have that $K_j$ is a clique on $d-i+2$ vertices, and for each vertex $v\in K_j$, $v$ has exactly $i-1$ neighbors in $V(G_{i-1})\setminus K_j$. Now to obtain the graph $G_i$, we first replace each clique $K_j$ in $G_{i-1}$ by its degree-preserving lift $L(K_j)$, where $j\in \{1,2,\ldots,l_{i-1}\}$.  In effect, each vertex $v$ of $G_{i-1}$ is replaced by a 
	$d-i+1$ sized clique $K'_v$ consisting of vertices, say $v_1,v_2, \ldots, v_{d-i+1}$.
	The vertex set of $G_i$ is the union of the vertex sets of all these $n_{i-1}$ cliques.  Clearly,
	$n_i = n_{i-1}. (d-i+1)$ and $V(G_i)$ has a trivial partition into cliques of order $d-i+1$,
	namely $\{ K'_v:  v \in V(G_{i-1}) \}$.  Now for each cross edge $uv$ of $G_{i-1}$ we add a
	perfect matching between $K'_v$ and $K'_u$ by making $v_j$ adjacent to $u_j$ for $1 \le j \le d-i+1$.
	Since there are exactly $i-1$ cross edges incident on each vertex $v$ in $G_{i-1}$, 
	each $v_j$ ($1 \le j \le d-i+1$) inherits $i-1$ cross edges in $G_i$ from the corresponding cross
	edges of $G_{i-1}$ incident on $v$. In addition to that, inside $L(K_j)$  (where $v \in K_j$) there is one cross edge incident at each $v_j$, along with the $d-i$ clique edges from the clique $K'_v$. Thus at a vertex $v_j$ of $G_i$ there are exactly $i$ cross edges and $d-i$ clique edges. 
	Now it is easy to see that  $G_i$  satisfies Property~\eqref{graphproperty}, for each $i\in \{1,2,\ldots,q-1\}$.


	Consider the graph $G_{q-1}$. Let $K_1, K_2,\ldots, K_{l_{q-1}}$ denote the partition of $V(G_{q-1})$ into cliques, where $l_{q-1}=\frac{n_{q-1}}{d-q+2}$. Since $G_{q-1}$ satisfies Property~\eqref{graphproperty}, we have that each vertex in any clique $K_j$, where $j\in \{1,2,\ldots,l_{q-1}\}$ has $q-1$ neighbors outside the clique. This implies that the set of all cross edges $M$ of $G_{q-1}$  form a $(q-1)$-regular subgraph of $G_{q-1}$. More precisely, $M=\{uv\in E(G_{q-1}):u\in K_i, v\in K_j$, where $i\neq j$, $i,j\in \{1,2,\ldots,l_{q-1}\}\}$. But it is easy to see that removing all the $\frac{n_{q-1}(q-1)}{2}$ edges of the subgraph $M$ of $G_{q-1}$ leaves $l_{q-1}=\frac{n_{q-1}}{d-q+2}$ connected components namely, $K_1,K_2,\ldots, K_{l_{q-1}}$. Coloring $\frac{n_{q-1}(q-1)}{2}$ edges in the subgraph $M$ with $\frac{n_{q-1}(q-1)}{2}$ distinct colors and coloring the edges of each of the $l_{q-1}=\frac{n_{q-1}}{d-q+2}$ components, with a new color yields a $q$-coloring using $\frac{n_{q-1}(q-1)}{2}+\frac{n_{q-1}}{d-q+2}$ colors. This
	implies that for the graph $G_{q-1}$, $\mathrm{OPT}(G_{q-1}) \ge \frac{n_{q-1}(q-1)}{2}+\frac{n_{q-1}}{d-q+2}$.
	\medskip
	
	On the other hand, it is easy to see that $M$ is not the only $(q-1)$-regular subgraph available in the graph $G_{q-1}$. There is another simple way to get a $(q-1)$-regular subgraph of $G_{q-1}$ as follows: Since $d\geq 2q+1$ and $d-q+2$ is even, from each clique $K_j$, where $j\in \{1,2,\ldots,l_{q-1}\}$, we can choose a $(q-1)$-regular subgraph, say $C_j$ using only the clique edges in such a way that $K_j\setminus C_j$ remains connected. The union of all these $(q-1)$-regular subgraphs, say $M'=\bigcup_{i=1}^j C_j$ is again a $(q-1)$-regular subgraph of $G_{q-1}$. Note that Algorithm~\ref{alg:feng} when applied to $G_{q-1}$, picks up an arbitrary maximum $(q-1)$-sub-factor of $G_{q-1}$, colors its edges with distinct colors, and then gives new distinct colors to each of the connected components that results when the chosen maximum $(q-1)$-sub-factor is removed from $G_{q-1}$. Note that both the subgraphs $M$ and $M'$ defined above are maximum $(q-1)$-sub-factors of $G_{q-1}$. Suppose that the Algorithm~\ref{alg:feng} picks up $M'$ instead of $M$ to start with. It is obvious that if $M'$ is removed from $G_{q-1}$, the resulting graph has only one connected component in it. Therefore, Algorithm~\ref{alg:feng} may yield a $q$- edge coloring of size $\frac{n_{q-1}(q-1)}{2}+1$. 
	
	\medskip
	
	Thus for $G_{q-1}$, 
	
	\begin{align*}
		\frac{{\rm OPT}(G_{q-1})}{{\rm ALG}(G_{q-1})}&\geq \frac{\frac{n_{q-1}(q-1)}{2}+ \frac{n_{q-1}}{d-q+2}}{\frac{n_{q-1}(q-1)}{2}+1}\\
		&=\frac{\big(\frac{n_{q-1}(q-1)}{2}+1\big)\bigg(1+\frac{2}{(q-1)(d-q+2)}-\frac{1+\frac{2}{(q-1)(d-q+2)}}{\frac{n_{q-1}(q-1)}{2}+1}\bigg)}{\frac{n_{q-1}(q-1)}{2}+1}\\
		&=\bigg(1+\frac{2}{(q-1)(d-q+2)}-\frac{1+\frac{2}{(q-1)(d-q+2)}}{\frac{n_{q-1}(q-1)}{2}+1}\bigg)
	\end{align*}
	The above ratio is very close to $\bigg(1+\frac{2}{(q-1)(d-q+2)}\bigg)$ for large $n_{q-1}$.
	
	\begin{remark}
		Note that in the proof of Theorem~\ref{thm:factor}, we bound the number of
		edges in the characteristic subgraph (which is the same as the number of colors
		in the coloring) by $\frac{n(q-1)}{2}+ \frac{n_{q}-n_{q-2}-2n_{q-3}-3n_{q-4}-\cdots-(q-1)n_{0}}{2}$. In the case of graphs with a $(q-1)$ factor, $\frac{n(q-1)}{2}+1$ is a lower bound on the number of colors returned by 
		Algorithm \ref{alg:feng}, and thus intuitively, the term $\frac{n_{q}-n_{q-2}-2n_{q-3}-3n_{q-4}-\cdots-(q-1)n_{0}}{2}$ was
		the ``excess''. We tried only  to tackle  this excess in the previous proof leaving the first part (i.e. $\frac {n(q-1)}{2}$, untouched). However, in the \emph {general case,}  $\frac {n(q-1)}{2}$ could be a gross
		overestimate of the number of edges in  a maximum $(q-1)$-sub-factor, so the previous strategy does not give any
		approximation guarantee. Therefore, to deal with general graphs, in the upcoming sections, we introduce a new parameter, $\kappa = \frac{n(q-1)}{2|E(H_{q-1})|}$, where $|E(H_{q-1})|$ is the number of edges in a maximum $(q-1)$-sub-factor of the given graph $G$.  
	\end{remark}
	\section{ Result for general graphs when $q=2$ } \label{sec:q=2}
	In this section, we  deal with the case $q=2$, and derive a bound for general graphs in terms of $\kappa$. Refer Section~\ref{results}, for the motivation behind the choice of parameter $\kappa$. First, we have the following lemma. 
	\begin{lemma}\label{lem:chargraph}
		For a graph $G$ with minimum degree, $\delta(G)\geq 3$,  and an edge $2$-coloring ${\cal C}$ of $G$, there exists a characteristic
		subgraph $\chi$ of $G$  such that $\chi$ 
		is a disjoint union of paths. 
	\end{lemma}
	\begin{proof}
		Let $\chi$ be a characteristic subgraph of $G$ (i.e. a subgraph of $G$ that contains exactly one of each color in ${\cal C}$), but
		with the minimum possible number of cycles. We claim
		that then  $\chi$ has no cycles. For the sake of contradiction, suppose that $\chi$ has a
		cycle. Let $u$ be one of the vertices in the cycle, and $v,w$ be its 
		two neighbors in the cycle. Since $\delta(G) \geq 3$,
		there is a  neighbor $z$ of $u$ in $G$, where $z \not \in \{v,w\}$.
		Now the edge $uz$ must have the same color as $uv$ or
		$uw$, say $uv$. It follows that $z$  
		is  incident with at most one  edge in
		$\chi$, i.e.  $degree_{\chi} (z) \le 1$, since otherwise the number of distinct colors incident at vertex $z$ is greater than or equal to 3,
		contradicting the assumption.  
		Then $\chi-uv+uz$ is a characteristic subgraph with a smaller
		number of cycles, violating the assumption of minimality. 
		Hence, $\chi$ is acyclic and is a disjoint union of paths.  
	\end{proof}
	In view of Lemma~\ref{lem:chargraph}, in this section, we can assume that $\chi$ is an acyclic subgraph of $G$ and the components of the characteristic
	subgraph (which are paths) will be called {\em characteristic paths}.
	We now have the following theorem.
	
	\begin{theorem}\label{thm:main}
		Let $G$ be an $n$-vertex  connected graph with minimum degree, $\delta(G)=\delta\geq 3$. Let $M$
		be a maximum matching of $G$. Then $ar(G,K_{1,3})\le |M| \big(1+ \frac {2(\kappa+1)} {\delta-1}\big)$, where $\kappa=n/2|M|$ being the ratio of
		number of vertices to the size of a maximum matching of the graph.
	\end{theorem}
	
	\begin{proof}
		Let ${\cal C}$ be an optimal coloring of $G$ using $c$ colors. Let $\chi$
		be a characteristic subgraph of $G$ with respect to the coloring ${\cal C}$,
		with the maximum number of characteristic paths. We will say a vertex $v$
		is an {\em internal} vertex of $\chi$ if it is an internal vertex of one of the
		characteristic paths. Similarly, a vertex will be called a {\em terminal}
		vertex of $\chi$ if it is a terminal vertex of one of the characteristic
		paths. Now, we pick a matching
		$M'$ from within $\chi$ by selecting alternate edges in each
		characteristic path, starting with the first edge in each path. Let $t$ be
		the number of {\em unselected} edges. Then we get:
		
		\begin {equation}
		\label {color-equation}
		c=|M'|+t\leq |M|+t 
		\end {equation}

		The remainder of the proof attempts to upper bound the excess term $t$. In
		fact we show that $t\leq |M|\cdot(2(\kappa+1)/(\delta-1))$ from which the
		theorem immediately follows. Let $N_2$ denote the set of all internal vertices in $\chi$ and let $T\subseteq N_2$
		be the set of vertices consisting of {\em left} endpoint of each unselected
		edge (see Figure \ref{fig:fig1}). Thus we have $|T|=t$. First, we note that $T$ is an independent set in
		$G$. This is because vertices in $T$ are mutually non-adjacent internal
		vertices of $\chi$ and hence have mutually disjoint incident colors. For
		each $v\in T$, choose a set of $\delta-2$ edges incident at $v$, which are
		not present in $\chi$ (this is possible, as each vertex has at most two
		neighbors in $\chi$ and $\delta(G)\geq 3$). Let us call these edges as {\em special} edges. Let $S_0,S_1$ and $S_2$ be sets of vertices of
		$V\backslash T$ which are incident with $0$, $1$ and $2$ of these special
		edges. Let $s_i=|S_i|$. Since the vertices in $T$ are incident with mutually disjoint sets
		of colors, a vertex in $V\backslash T$ is incident with at most two special
		edges emanating from $T$. Thus $V\backslash T = S_0\uplus S_1\uplus S_2$,
		or $s_0+s_1+s_2=n-t$.
		Counting the special edges across the bipartition $(T, V\backslash T)$ we
		have:
		\begin{equation*}
			t(\delta-2) = 2s_2+s_1 
		\end{equation*}
		Moreover, as $s_0+s_1+s_2=n-t$, we can rewrite the above equality as
		$t(\delta-2)=n-t + (s_2-s_0)$, and thus,
		\begin{equation}\label{eq:h2}
			t(\delta-1)\leq n + s_2.
		\end{equation}

		\begin{figure}[h]
			\centering
			\begin{tikzpicture}[xscale=1.2]
				\foreach \i in {0,1,2,3} {%
					\draw (0,\i) -- (4,\i);
					\foreach \j in {0,1,2,3,4} {%
						\draw (\j,\i) node {$\bullet$};
					}
				}
				
				\foreach \i in {0,1,2,3} {%
					\foreach \j in {1,3} {%
						\draw (\j,\i) node[draw, fill=none, rectangle] {$\bullet$};
						\draw (\j+0.5,\i) node {$\times$};
					}
				}
				
				\draw (6,1.5) ellipse (1cm and 2cm);
				Put some nodes in the ellipse
				\coordinate (A) at (6,2.5);
				\coordinate (B) at (6.5, 1.5);
				\coordinate (C) at (5.5, 0.5);
				
				\draw (A) node {$\bullet$};
				\draw (B) node {$\bullet$};
				\draw (C) node {$\bullet$};
				
				\draw[dashed] (A) -- (1,3) (A) -- (3,1)
				(B) -- (1,3) (B) -- (3,2)
				(C) -- (1,0) (C) -- (3,0);
				
			\end{tikzpicture}

			\caption{Characteristic paths are shown on the left. The unselected edges
				are marked with $\times$, and corresponding vertices in $T$ are indicated
				by a box. The special edges are indicated with dashed lines. The set $S_2$
				is indicated by the ellipse.}
			\label{fig:fig1}
		\end{figure}
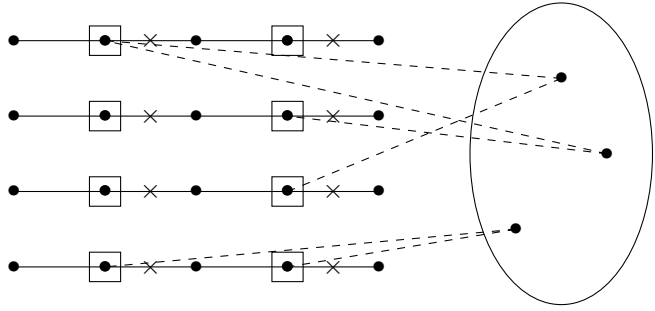

		To obtain a bound on $s_2$, we consider the neighbors of vertices in $S_2$
		which are not in $T$. Note that each vertex $v\in S_2$ has at least
		$\delta-2$ neighbors outside $T$. Let $N_1$ be the set of terminal vertices
		of the characteristic subgraph $\chi$ and $N_2'\subseteq N_2$ be the set of internal
		vertices of $\chi$, which are not in $T$. Let $N_0$ denote the set of
		vertices that are not incident with an edge in $\chi$. Note that
		$S_2\subseteq N_0$. This is because a vertex incident with an edge in $\chi$ 
		can only receive special edges of at most one  color. But then it has
		at most one neighbor in $T$, and hence it is not a vertex in $S_2$. We show the following:
		
		~~~~~~
		\begin{myclaim} \label{acoloredge}
			Let $a$ be a color present in a characteristic
			path of length at least two in $\chi$. Then there is no $a$-colored edge
			between two vertices of $N_0$.
		\end{myclaim}
		
		\smallskip
		\noindent {\it Proof of Claim:}  We prove by contradiction. Let $uv$ be an edge
		in a characteristic path of length at least two, and let $a$ be the color of
		$uv$. If possible, let
		$x$ and $y$ be two vertices in $N_0$ with edge $xy$ having color $a$. Then
		$\chi-uv+xy$ is a characteristic subgraph with more characteristic paths than $\chi$, a
		contradiction. The claim follows.
		
		~~~~
		\begin{myclaim}\label{acolortriangle}
			Let $a$ be a color present in a characteristic
			path of length at least two in $\chi$. Then there is no $a$-colored
			triangle formed by the $a$-colored edge in $\chi$ and a vertex from $N_0$.
		\end{myclaim}
		\smallskip
		\noindent {\it Proof of Claim:} To
		prove, again assume that $uv$ is an $a$-colored edge in $\chi$ where $u$ is
		not a terminal vertex. If possible, let $uvw$ be an $a$-colored triangle
		with $w\in N_0$. Again, $\chi-uv+vw$ is a characteristic subgraph with more
		characteristic paths than $\chi$, contradicting the choice of $\chi$. The claim
		follows.
		
		~~~~
		\begin{myclaim} \label{N(S_2)}
			The neighbors of vertices in $S_2$, which are not in $T$, lie in $N_1$.
		\end{myclaim}
		
		\smallskip
		\noindent {\it Proof of Claim:} To prove, we observe that each color incident on vertices in $S_2$ appears in some
		characteristic path of length at least two (as vertices in $T$ are on such
		paths). Then as $S_2\subseteq N_0$, from Claim~\ref{acoloredge}, we have that a vertex $v\in S_2$ is not adjacent to a vertex in $N_0$.
		Further, $v$ is not adjacent to a vertex in $N_2'$: Suppose $v$ is adjacent
		to a vertex $u$ in $N_2'$.  Let $color(vu) = a$. 
		Since $u \in N_2'$, it is an internal vertex in $\chi$, and therefore 
		there should be an edge, say $uw$ in $\chi$  such that $color(uw)=a$. Now there should be an edge $vx$
		such that $color(vx)=a$,   for some $x \in T$ since $v \in S_2$.
		But it is clear that $uw$ is the only edge colored with color $a$ in $\chi$,
		and since $u \not \in T$, $w=x$.  But then $v$ is adjacent to both
		$u$ and $w$, forming an $a$-colored triangle containing an edge of $\chi$,
		which is a contradiction by Claim~\ref{acolortriangle}. We infer that $v$ cannot be 
		adjacent to any vertex in $N_2'$. 
		Hence, all the neighbors of
		$v$, not in $T$ are in $N_1$.
		
		~~~~~~~~~
		
		Consider the bipartite graph $X$ with bipartition $(S_2, N_1)$ with the edge set
		consisting of $N_1$-$S_2$ edges in $G$. Clearly $d_X(v)\geq \delta-2$ for $v\in
		S_2$. We now prove that $d_X(u)\leq \delta-2$ for $u\in N_1$. Let $u\in N_1$,
		and let $uz$ be the edge of $\chi$ incident at $u$, where $a$ is the color
		of $uz$. Let $v$ be a neighbor of $u$ in $S_2$. We show that $uv$ is not of
		color $a$. For the sake of contradiction, suppose that color of $uv$ is $a$.
		Then as $v\in S_2$, it has a neighbor $w$ in $T$, which is incident with
		an edge of color $a$. But then $w=z$ as $z$ is the only  possible 
		internal vertex in
		$\chi$ incident with edge of color $a$. Now $uvz$ form an $a$-colored
		triangle, which contradicts Claim~\ref{acolortriangle}. Thus, all the edges of $X$ incident
		at $u$ must have the color different from $a$ and hence must have the same
		color (say $b$). Clearly, all the vertices in $S_2$ incident with an edge of
		color $b$ must be incident with a $b$-colored special edge from a vertex
		in $T$. As vertices in $T$ are incident with mutually disjoint colors, we
		infer that all such vertices are incident with $b$-colored special edges
		from a single vertex $w\in T$ (see Figure \ref{fig:fig2}). Since there are exactly $\delta-2$ special
		edges emanating from a vertex in $T$, we conclude that there are, at most
		$\delta-2$ vertices in $S_2$ incident with color $b$. Hence $d_X(u)\leq
		\delta-2$ for $u\in N_1$. Now, we have $s_2(\delta-2)\leq |E(S_2,N_1)|\leq
		(\delta-2)\cdot |N_1|$. Finally, observe that $|N_1|\leq 2|M|$ and hence
		$s_2\leq 2|M|$. Substituting in Equation (\ref{eq:h2}), we have:
		\begin{align}
			t(\delta-1) &\leq n + 2|M| \nonumber \\
			&\leq  |M|\cdot\left(\frac{n}{|M|}+2\right) \nonumber\\
			t&\leq 2|M|\cdot\left(\frac{\kappa+1}{\delta-1}\right) \label{eqn:kappaq=2} \hspace{1cm} \text { \big(as } \kappa=\frac{n}2{|M|}\big).
		\end{align}

		\medskip
		
		Substituting \eqref{eqn:kappaq=2} in inequality \eqref {color-equation}, we then have 
		
		\begin{align*}
			ar(G,K_{1,3})=c\le |M| \Big(1+ \frac {2(\kappa+1)} {\delta-1}\Big)
		\end{align*}
		This proves the theorem.
	\end{proof}
	\smallskip
	
	The corollary below follows from the fact that, when $q=2$, for any graph $G$, $\mathrm{OPT}(G)=ar(G,K_{1,3})$ and $\mathrm{ALG}(G)\geq |M|$, where $M$ is a maximum matching of $G$.
	\begin{corollary}
		Let $G$ be an $n$-vertex  connected graph with $\delta(G)\geq 3$. Let $M$
		be a maximum matching of $G$. Then $\rm {OPT(G)}\le  \Big(1+ \frac {2(\kappa+1)} {\delta-1}\Big)\rm {ALG(G)}$, where $\kappa=n/2|M|$ being the ratio of
		number of vertices to the size of a maximum matching of the graph.  
	\end{corollary}
	
	\begin{figure}[h]
		\centering
		\begin{tikzpicture}[xscale=1.2]
			
			\draw (0,1) ellipse (1cm and 2.5cm);
			\draw (3,1) ellipse (1cm and 2.5cm);
			\draw (3,1.5) ellipse (0.5cm and 1cm);
			\draw (5,5) ellipse (1.5cm and 0.6cm);
			
			\coordinate (U) at (0.5,1);
			\coordinate (U1) at (-1.2, 1.5);
			\coordinate (V1) at (3.2,1.7);
			\coordinate (V2) at (2.8,1.2);
			\coordinate (V3) at (2.8,2.0);
			\coordinate (T1) at (5,5);
			
			\draw (U) node {$\bullet$};
			\draw (U) node[below] {$u$};
			\draw (V1) node {$\bullet$};
			\draw (V1) node[right] {$v_1$};
			\draw (V2) node {$\bullet$};
			\draw (V2) node[right] {$v_2$};
			\draw (V3) node {$\bullet$};
			\draw (V3) node[right] {$v_3$};
			\draw (T1) node {$\bullet$};
			\draw (T1) node[above] {$w$};
			\draw (U1) node {$\bullet$};
			\draw (U1) node[below] {$z$};
			
			\draw (U1) to node[above] {$a$} (U);
			\draw (U) to node[above] {} (V1);
			\draw (U) to node[below] {$b$} (V2);
			\draw (U) to node[above] {$b$} (V3);
			\draw[dashed] (T1) to node[right] {$b$} (V1);
			\draw[dashed] (T1) to node {} (V2);
			\draw[dashed] (T1) to node[left] {$b$} (V3);
			
			\coordinate (A) at (0,3.5);
			\coordinate (H) at (3,3.5);
			\coordinate (T) at (5,5.5);
			
			\draw (A) node[above] {$N_1$};
			\draw (H) node[above] {$S_2$};
			\draw (T) node[above] {$T$};
			
		\end{tikzpicture}
		\caption{The vertex $u$ in $N_1$ is incident with $a$-colored edge in $\chi$.
			Its neighbors in $S_2$ are through the ``other'' color at $u$, namely $b$. Moreover,
			these neighbors are also incident with $b$-colored special edges from a
			vertex in $T$.}
		\label{fig:fig2}
	\end{figure}
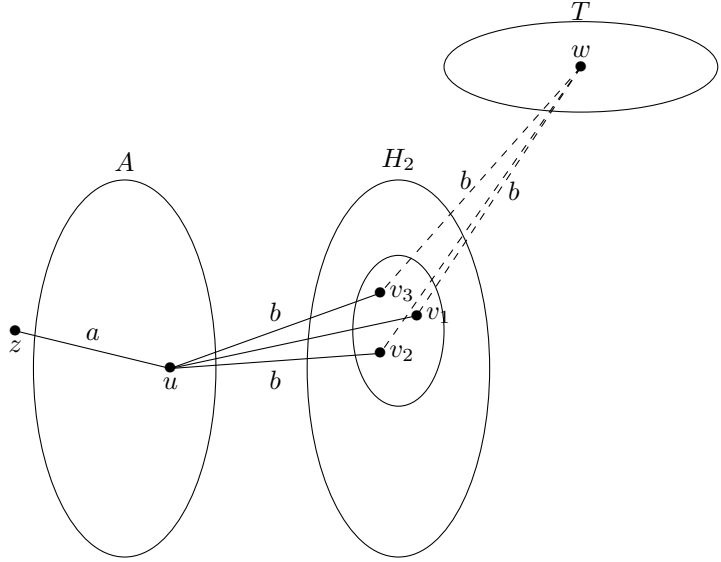
	
	\section{Result for general graphs when $q>2$}
	As a generalization of the results obtained in Section~\ref{sec:q=2}, here we derive a bound for general graphs (for $q>2$) in terms of $\kappa$. Even though we achieve this by extending some ideas that we used in proving Theorem~\ref{thm:main}, here it is crucial to have $q$ to be a positive integer such that `$q>2$'. 
	\begin{theorem} \label{thm:main2}
		Let $G$ be a connected graph on $n$ vertices, with minimum degree $\delta(G)=\delta\geq q+1$, where $q$ is a positive integer such that $q>2$. Let $E(H_{q-1})$ be the set of edges in a maximum $(q-1)$-sub-factor of $G$. Then $ar(G,K_{1,q+1}) \le|E({H_{q-1})}|\bigg( 1+ \frac {2\big(\kappa+\frac{1}{(q-1)}\big)} {\delta-1}\bigg)$, where $\kappa= \frac{n(q-1)}{2|E(H_{q-1})|}$.
	\end{theorem}
	\begin{proof}
		Let $\mathcal{C}$ be an optimal edge $q$-coloring of $G$ using $c$ colors. Let $\chi$ be a characteristic subgraph of $G$ with respect to the coloring $\mathcal{C}$, having a minimum  possible number of $q$ degree vertices. Let $T$ denote the set of $q$ degree vertices in $\chi$, where $|T|=t$. 
		Let $\chi'$ denote the subgraph of $G$ obtained from $\chi$ by removing one edge incident to each of the $t$ vertices in $T$. Clearly, the maximum degree of $\chi'$ is less than or equal to $q-1$. Thus we have, 
		\begin{equation} \label{coloreqn}
			c=|E(\chi)|\leq |E(\chi')|+t  \leq |E(H_{q-1})|+t
		\end{equation}
		where $|E(H_{q-1})|$ is the number of edges in a maximum $(q-1)$-sub-factor of $G$.
		In the remainder of the proof, we bound this excess $t$. Our goal is to show that $t\leq 2|E(H_{q-1})|\Big(\big(\kappa+\frac{1}{q-1}\big)/(\delta-1)
		\Big)$.
		
		\medskip
		
		For each $v\in T$, choose a set of $\delta-q$ edges incident at $v$, which are not present in $\chi$ (this is possible since each vertex in $T$ has at most $q$ neighbors in $\chi$ and $\delta \geq q+1$). As in Theorem~\ref{thm:main}, we call these edges as \emph{special edges} and say that $u$ is a \emph{special neighbor} of $v$ if $uv$ is a special edge in $G$. Note that for any vertex $v\in T$, any special neighbor of $v$ belong to the set $V\setminus T$. This is because, any edge between a pair of vertices in $T$ has to be an edge in $\chi$. Recall that for each $i\in \{0,1,2,\ldots,q\}$, $N_i$ denote the set of vertices that has degree $i$ in $\chi$. Clearly, $T=N_q$ and $V\setminus T= \biguplus _{i=0}^{q-1} N_i$. Then by the definition of special neighbors and the fact that each vertex in $G$ has at most $q$ distinct colors incident to it, we have the following remark. 
		\begin{remark} \label{rem:specialcolor}
			Let $xu$ be a special edge in $G$ having a color $a$ on it such that $x\in T$ and $u\in V\setminus T$. Then the color $a$ is present in $\chi$ on an edge incident to $x$, say $xw$ and $u\neq w$.    
		\end{remark}
		
		\medskip
		Now we claim the following.
		\medskip
		\begin{myclaim} \label{claim:samecolor}
			\emph{Let $u,v\in V\setminus T$ (possibly, $u=v$) and $x,y\in T$ be such that $x\neq y$, and the special edges $ux$ and $vy$ incident to $u$ and $v$ have the same color with respect to the coloring $C$. Then $u,v\in N_{q-1}$. Moreover, the vertices $x$ and $y$ are the only possible special neighbors of both the vertices $u$ and $v$.}
		\end{myclaim}
		
		\smallskip
		\noindent {\it Proof of Claim:}
		Let $a$ be the color on the edges $ux$ and $vy$. Since $x,y\in T$ and $x\neq y$, by Remark~\ref{rem:specialcolor}, we have that $xy$ is an edge in $\chi$ and $a$ is the color on the edge $xy$. Without loss of generality, assume that $u\in N_i$, for some $i\leq q-2$. Then, as $d_{\chi}(u)\leq q-2$, we have that $\chi- xy + ux$ is a characteristic graph of $G$ having a fewer number of degree $q$ vertices than $\chi$. This contradicts the choice of $\chi$. Since the proof for the case, $v\in N_{i}$ is symmetric, we can therefore conclude that $u,v\in N_{q-1}$. 
		
		Suppose that there exists a vertex $z\neq x,y$ in $T$ such that $z$ is a special neighbor of $u$ (as in the previous paragraph, proof for the case in which \textit{$z$ is a special neighbor of $v$} is symmetric). Let $b$ be the color on the special edge $uz$. Note that $a\neq b$, as $z\neq x,y$, and at most two vertices in $T=N_q$ can share the same color $a$. Now as $z\in T$, we have by Remark~\ref{rem:specialcolor} that the color $b$ is present in $\chi$ on an edge incident to the vertex $z$, say $zw$ and $w\neq u$. Since $u\in N_{q-1}$, there are $q-1$ distinct colors incident at $u$ in $\chi$. Then by the definition of special edges and Remark~\ref{rem:specialcolor}, we have that all these $q-1$ colors are distinct from both the colors $a$ and $b$ on the special edges incident at $u$ . Since $a\neq b$, we then have $q+1$ distinct colors incident with the vertex $u$. This contradiction proves the claim. 
		
		\bigskip
		\begin{myclaim}\label{claim:atmostq}
			\emph{Each vertex $u\in V\setminus T$ has at most $q$ special edges incident to it.}
		\end{myclaim}
		\smallskip
		\noindent {\it Proof of Claim:}
		Suppose not. Assume that there exists a vertex $u\in V\setminus T$ that has at least $q+1$ special edges incident to it. Since there is only $q$ distinct colors incident at $u$, there exist at least two vertices $x,y\in T$ such that the special edges $ux$ and $uy$ have the same color. By the application of Claim~\ref{claim:samecolor} (with $u=v$), we then have that $u$ cannot have more than two special neighbors in $T$ and, therefore, cannot have more than two special edges incident to it. This contradicts the fact that $q+1>2$. This proves the claim.
		
		\bigskip
		Let $S_0,S_1,S_2,\ldots,S_q$ denote the set of vertices in $V\setminus T$, which are incident with $0,1,2,\ldots,q$ special edges incident to it. By Claim~\ref{claim:atmostq}, we have $V\setminus T=\biguplus _{i=0}^{q} S_i$. Note that the definition of special edges together with Remark~\ref{rem:specialcolor} implies that for any $u\in S_i$, where $i\in \{1,2,\ldots,q\}$, the color on the  special edges incident at $u$ is distinct from the color on the edges incident at $u$ in $\chi$.
		We then have the following remark due to Claim~\ref{claim:samecolor} and the fact that each vertex in $G$ has at most $q$ distinct colors incident to it.
		\begin{remark} \label{disjointcolors}
			\begin{enumerate}
				\vspace{-0.06in}
				\renewcommand{\theenumi}{\alph{enumi}}
				\renewcommand{\labelenumi}{(\textit{\theenumi})}
				\item \label{rem:degree} For each $i>2$, the $i$ special edges incident to each vertex in $S_i$ are of distinct colors. Further, for each $i> 2$ and $u\in S_i$, we have that $d_{\chi}(u)\leq q-i$ (i.e. for each $i>2$, we have $S_i\subseteq N_0\cup N_1\cup\ldots\cup N_{q-i}$).
				\item \label{rem:distinctT} Let $i>2$ and $x,y\in T$. If there exist distinct vertices $u,v\in S_i$ such that the special edges $ux$ and $vy$ have the same color, then $x=y$.
			\end{enumerate}
			
		\end{remark}
		
		\medskip
		Let $s_i=|S_i|$. Then, as $V\setminus T=\biguplus _{i=0}^{q} S_i$, we have $s_0+s_1+\cdots+s_q = n-t$. Counting the special edges across the bipartition $(T,V\setminus T)$, we have:
		\begin{align}
			t(\delta -q) &= qs_q+(q-1)s_{q-1}+\cdots+2s_2+s_1 \nonumber\\ 
			& = (q-1)(s_q+s_{q-1}+\cdots+s_2+s_1+s_0)+s_q-s_{q-2}-2s_{q-3}-\cdots-(q-1)s_0 \nonumber
		\end{align}
		Since $s_0+s_1+\cdots+s_q = n-t$, the above equation can be written as 
		
		\begin{align*} 
			t(\delta -q) &=  
			(q-1)(n-t)+s_q-s_{q-2}-2s_{q-3}-\cdots-(q-1)s_0 \\
			&\leq (q-1)(n-t)+s_q 
		\end{align*}
		Thus we have,
		\begin{align}\label{eqn:main}
			t(\delta -1) &\leq  
			n(q-1)+s_q 
		\end{align}
		To obtain a bound on $s_q$, we consider the neighbors of vertices in $S_q$ that are not present in $T$. Since $\delta\geq q+1$, each vertex $u\in S_q$ has at least $\delta -q$ neighbors outside $T$. (For any vertex $u\in S_q$, it is not difficult to verify that all the neighbors of $u$ in $T$ are its special neighbors and, hence $u$ has at most $q$ neighbors in $T$). Also, note that as $q>2$, by Remark~\ref{disjointcolors}(\ref{rem:degree}), we have that each vertex in $S_q$ has zero degree in $\chi$. Therefore, $S_q\subseteq N_0$.
		
		\medskip
		Let $B=N_0\cup N_1\cup \cdots \cup N_{q-2}$. We then have the following claim. 
		
		\bigskip
		\begin{myclaim}\label{claim:nbrHq}
			\emph{For any pair of vertices $u\in B$ and $v\in S_q$ we have $uv\notin E(G)$.}
		\end{myclaim}
		\smallskip
		\noindent {\it Proof of Claim:}
		Suppose not. Assume that there exist vertices $u\in B$ and $v\in S_q$ such that $uv\in E(G)$. Let $a$ be the color on edge $uv$ with respect to the coloring $\mathcal{C}$. Since $v\in S_q$, by Remark~\ref{disjointcolors}(\ref{rem:degree}), there exists a vertex $w\in T$ such that the special edge $vw$ has the color $a$ on it. Since $w\in T=N_q$, there exists an edge say $wx$ in $\chi$ that has the color $a$ on it (possibly, $x=u$). Note that as $v\in S_q\subseteq N_0$ and $u\in B= N_0\cup N_1\cup \cdots \cup N_{q-2}$, degree of both the vertices $u$ and $v$ are at most $q-2$ in $\chi$. Therefore, $\chi-wx+uv$ is a characteristic subgraph of $G$ that has a fewer number of $q$ degree vertices than $\chi$. This contradicts the choice of $\chi$ and proves the claim. 
		
		\bigskip
		
		Therefore by Claim~\ref{claim:nbrHq} and the fact that $V\setminus T = B\biguplus N_{q-1}$, we can conclude that each vertex in $S_q$ has its $\delta-q$ neighbors in $N_{q-1}$. Now consider the bipartite graph $X$ with bipartition $(S_q,N_{q-1})$ with edge set consisting of $S_q-N_{q-1}$ edges in $G$. Clearly, for each vertex $v\in S_q$, we have $d_X(v) \geq \delta-q$. We now prove that for each vertex $u\in N_{q-1}$, $d_X(u) \leq \delta-q$. Let $u\in N_{q-1}$. Let $a$ be any color among the $q-1$ colors on the edges incident with $u$ in $\chi$. Let $uz$ be the edge colored $a$ in $\chi$. Clearly, $z\notin N_0$ and therefore $z\notin S_q$ (since $S_q\subseteq N_0$). We claim that the color $a$ is not present on any edge of the form $uv$, where $v\in S_{q}$. Suppose not. Let $v\in S_{q}$ be such that the edge $uv$ is colored $a$ with respect to the coloring $\mathcal{C}$. Since $v\in S_q$, we then have that there exists a  vertex $w\in T=N_q$ such that the special edge $vw$ has the color $a$ on it. Thus by choice of the color $a$ and Remark~\ref{rem:specialcolor}, it should be the case that $w=z$. But then as $d_{\chi}(v)=0$, we can observe that $\chi-uz+uv$ is a characteristic subgraph of $G$ that have fewer number of $q$ degree vertices than $\chi$. This contradicts the choice of $\chi$. Therefore, we can conclude that for any vertex $u\in N_{q-1}$, none of the $q-1$ colors on the edges incident to $u$ in $\chi$ is present on the edges of the form $uv$, where $v\in S_{q}$. This implies that all the edges of $X$ incident
		at $u$ must have the same color (which is remaining at $u$), say $b$. Since $q>2$, by Remark~\ref{disjointcolors}(\ref{rem:degree}), we have that each vertex in $S_q$ is incident with $q$ special edges, all having distinct colors on it. This implies that all the vertices in $S_q$ that are incident with a $b$-colored edge in $X$ are also incident with a $b$-colored special edge. Again, as $q>2$, by Remark~\ref{disjointcolors}(\ref{rem:distinctT}), we can conclude that all the $b$ colored special edges incident at the vertices in $S_q$ have its other end point, a single vertex from $T$, say $w\in T$ (refer Figure~\ref{fig:fig2}, where in this context, $N_1$ and $S_2$ are now replaced by $N_{q-1}$ and $S_q$ respectively). Since there are exactly $\delta -q$ special edges emanating from a vertex in $T$, this implies that there are at most $\delta-q$ vertices in $S_q$ incident with color $b$. Hence $d_X(u)\leq \delta-q$ for each $u\in N_{q-1}$. Thus we have, $s_q(\delta-q)\leq |E(S_q,N_{q-1})|\leq (\delta-q)|N_{q-1}|$. This implies that $s_q\leq |N_{q-1}|$.
		
		\medskip
		
		Applying the above inequality in \eqref{eqn:main}, we then have:
		\begin{align} \label{eqn:second}
			t(\delta -1) &\leq  
			n(q-1)+|N_{q-1}|   
		\end{align}
		
		\medskip
		Recall the $(q-1)$-degree bounded subgraph $\chi'$ that we have defined earlier. We have, 
		\begin{align}
			|E(H_{q-1})| \geq |E(\chi')| &\geq |E(\chi)|-t  \nonumber\\
			&\geq \frac{1}{2}\big[tq+(q-1)|N_{q-1}|\big]-t \nonumber\\ 
			\text{i.e. }\Big(\frac{q-1}{2}\Big)\big|N_{q-1}\big| &\leq \big|E(H_{q-1})\big|- \Big[\frac{tq}{2}-t\Big] \nonumber\\ 
			|N_{q-1}| &\leq \Big(\frac{2}{q-1}\Big)\big|E(H_{q-1})\big| \hspace{1.5cm} \text{(since $q>2$) \nonumber}
		\end{align}
		Substituting the above in \eqref{eqn:second}, we then have:
		\begin{align*}
			t(\delta -1) &\leq  
			n(q-1)+\Big(\frac{2}{q-1}\Big)\big|E(H_{q-1})\big|\\
			&= 2\big|E(H_{q-1})\big|\bigg[\frac{n(q-1)}{2\big|E(H_{q-1})\big|}+\frac{1}{q-1}\bigg]
		\end{align*}
		This implies that,
		\begin{align}\label{eqn:kappa}
			t&\leq 2\big|E(H_{q-1})\big|\bigg[\frac{\kappa+\frac{1}{q-1}}{\delta-1}\bigg] \hspace{1cm} \text { \bigg(as } \kappa=\frac{n(q-1)}{2|E(H_{q-1})|}\bigg) 
		\end{align}

		\medskip
		
		By substituting \eqref{eqn:kappa} in inequality \eqref{coloreqn}, and the fact that $ar(G,K_{1,q+1})=\mathrm{OPT(G)}=c$, we then have 
		
		\begin{align*}
			ar(G,K_{1,q+1})=c\leq |E(H_{q-1})|+t \leq |E(H_{q-1})|\bigg[1+\frac{2(\kappa+\frac{1}{q-1})}{\delta-1}\bigg].
		\end{align*}
		Hence the theorem.
	\end{proof}
	\smallskip
	As in the case $q=2$, we have the following corollary due to the fact that, ${\rm OPT}(G)=ar(G,K_{1,q+1})$ and ${\rm ALG}(G) \geq |E(H_{q-1})|$, where $E(H_{q-1})$ is the set of edges in a maximum $(q-1)$-sub-factor of $G$.
	\begin{corollary} \label{corr:main2}
		Let $G$ be a connected graph on $n$ vertices, with $\delta(G)\geq q+1$. Let $E(H_{q-1})$ be the set of edges in a maximum $(q-1)$-sub-factor of $G$. Then ${\rm OPT}(G)\leq \bigg( 1+ \frac {2\big(\kappa+\frac{1}{(q-1)}\big)} {\delta-1}\bigg){\rm ALG}(G)$, where  $\kappa= n(q-1)/2|E(H_{q-1})|$.
	\end{corollary}
	\subsection{Tightness for Theorem~\ref{thm:main} and Theorem~\ref{thm:main2}:} \label{sec:tightkappa}
	In this section, we propose a construction of a family of graphs for which the sub-factor based algorithm can end up with  an approximation guarantee almost equal to the expressions obtained in Theorem~\ref{thm:main} and Theorem~\ref{thm:main2}. 
	\begin{theorem}\label{thm:tightex1}
		Let a positive rational number $\kappa$ and a positive integer $q\geq 2$ be given. Then for any $\delta> \max\{q,(2\kappa-1)(q-1)\}$, it is possible to find infinitely many positive integers $t$ such that there exists a graph $G$ on $2\kappa t(q-1)$ vertices and with the following properties: (1) $|E(H_{q-1})|=t(q-1)^2$, (2) minimum degree of $G$ equals $\delta$, and (3) $G$ possesses an edge $q$-coloring using at least $|E(H_{q-1})|\Big(1+ \frac{2(\kappa-1)}{\delta-1}\Big)$ colors, where $E(H_{q-1})$ denotes the set of edges in a maximum $(q-1)$-sub-factor of $G$.
		
	\end{theorem} 
	
	\medskip
	\noindent{\textbf{Comment:}} Note that by Theorem~\ref{thm:main} and Theorem~\ref{thm:main2}, we have the following for each $q\geq 2$: For a connected graph $G$ with minimum degree $\delta(G)=\delta\geq q+1$, $ar(G,K_{1,q+1}) =\mathrm{OPT}(G)\le|E({H_{q-1})}|\bigg( 1+ \frac {2\big(\kappa+\frac{1}{(q-1)}\big)} {\delta-1}\bigg)$.  Now, for a graph $G$, satisfying the conditions in Theorem~\ref{thm:tightex1}, we have that
	$ar(G,K_{1,q+1}) =\mathrm{OPT}(G)\geq |E(H_{q-1})|\Big(1+ \frac{2(\kappa-1)}{\delta-1}\Big)$. This implies that, the approximation guarantee that we provided in Theorem~\ref{thm:main} and Theorem~\ref{thm:main2} are almost tight.
	
	\medskip
	\begin{proof}
		We will derive the value of $t$ in terms of $\kappa, q$, and $\delta$ towards the end of the proof. Let $G$ be a bipartite graph with partite sets $S$ and $T$, where the set $S$ is a disjoint union of $q-1$ sets say, $S_1, S_2,\ldots, S_{q-1}$ and $T$ is a disjoint union of $q$ sets, say $T_1,T_2,\ldots, T_{q-1}$ and $B$, where $|S_i|=|T_i|=t$, for each $i\in \{1,2,\ldots,q-1\}$ (refer Figure~\ref{fig:fig4}). 
		
		\medskip
		
		For $i\in \{1,2,\ldots,q-1\}$, $j\in \{1,2,\ldots,t\}$, let $s^j_i$ and $t^j_i$ respectively denote the $j^{th}$ vertices of the sets $S_i$ and $T_i$. In $G$, we first add some edges to form a $(q-1)$-regular subgraph, say $M$ as follows: For each $i\in \{1,2,\ldots,q-1\}$, $j\in \{1,2,\ldots,t\}$, make the vertex $s^j_i\in S_i$ adjacent to each vertex in the set, $\{t_1^j, t_2^j,\ldots, t_{q-1}^j\}$. Clearly, the edges of $M$ form a $(q-1)$-regular subgraph on $2t(q-1)$ vertices (since $|S\cup T|=2t(q-1)$). In our coloring of $G$, we give $t(q-1)^2$ distinct colors to  the edges of $M$. Now, we will add more edges to the graph $G$. Note that, irrespective of any additional edges in $G$, the subgraph $M$ remains as a suitable candidate for a maximum $(q-1)$-sub-factor, $H_{q-1}$ of $G$. This is because all the vertices in the partite set $S$ of the bipartite graph $G$ have degree $q-1$ in $M$. Thus we have, $|E(H_{q-1})|=|M|=t(q-1)^{2}$.  This proves condition~(1) in the statement of the theorem.
		
		\medskip
		
		For $i\in \{1,2,\ldots,q-1\}$, $j\in \{1,2,\ldots,t\}$, let $S_i=X_i\uplus Y_i\uplus Z_i$, where $X_i=(x_i^1,x_i^2,\ldots,x_i^{\delta-q+2})$, $Y_i=(y_i^1,y_i^2,\ldots,y_i^h)$ and $Z_i=(z_i^1,z_i^2,\ldots,z_i^{h'})$. Note that the value of $h$ will be derived along with the parameter $t$ in the later part of the proof, but we define $h'=t-(\delta-q+2)-h$.
		Now for each $i\in \{1,2,\ldots,q-1\}$, in $G$, we also make all the vertices in $X_i$ adjacent to all the vertices in $T_i\cup B$ (see the dashed edges in Figure~\ref{fig:fig4}). This implies that each vertex in $T_i$ has degree exactly $\delta$ and each vertex in the sets $X_i$ and $B$ has degree at least $\delta$ in $G$. Further, in the above coloring of $G$, we can give all these additional edges (i.e. dashed edges in Figure~\ref{fig:fig4}) a single color. 
		\medskip
		
		
		Let $B=\{v_1,v_2,\ldots,v_{\alpha}\}$, where $\alpha=(h'+1)(\delta-1)$. In fact, the set $B$ can be viewed as $(h'+1)$ bags, say $B_1,B_2,\ldots,B_{h'+1}$, where each bag is of size $(\delta-1)$. In $G$, for each $i\in \{1,2,\ldots,q-1\}$, we make all the vertices in the set $Y_i$ adjacent to all the vertices in $B_1$ so that each vertex in $Y_i$ has degree at least $\delta$ in $G$. In the coloring of $G$, we introduce $(q-1)$ new colors to color the sets of edges connecting $Y_i$ and $B_1$, where $i\in \{1,2,\ldots,q-1$\} (one color for each set). In $G$, for each $j\in \{1,2,\ldots,h'\}$, we also make the vertices in the set $\{z_1^j,z_2^j,\ldots,z_{q-1}^j\}$ adjacent to all the vertices in $B_{j+1}$, so that each vertex in $Z_i$ has degree at least $\delta$ in $G$ (see bold edges in Figure~\ref{fig:fig4}). Here we introduce $h'(q-1)$ new colors in the coloring of $G$ to color the sets of edges connecting the set $\{z_1^j,z_2^j,\ldots,z_{q-1}^j\}$ to each bag $B_{j+1}$, where $j\in \{1,2,\ldots,h'\}$. Therefore, we have an edge coloring of the graph $G$ that uses at least $t(q-1)^{2}+ (h'+1)(q-1)+1$ colors. Moreover, as we have the minimum degree of $G$ to be $\delta$, condition~(2) in the statement of the theorem is also satisfied.
		
		\medskip
		Recall that $|E(H_{q-1})|=|M|=t(q-1)^{2}$.  Then, by the definition, the value of $\kappa$ has to be, $\kappa=\frac{n(q-1)}{2|E(H_{q-1})|}=\frac{n(q-1)}{2t(q-1)^{2}}$. This implies that the number of vertices in $G$, i.e. $n=2\kappa t(q-1)$.
		
		\medskip
		Alternatively, from the construction of the graph $G$, we have the number of vertices in $G$ as follows: 
		\begin{align*}
			n &= 2t(q-1) + (t-(\delta-q+2)-h+1)(\delta-1)
		\end{align*}
		Therefore we have,
		\begin{align*}
			2\kappa t(q-1) &= 2t(q-1) + (t-(\delta-q+2)-h+1)(\delta-1)
		\end{align*}
		From the above equation, we can infer that the value of $h$ has to be the following:
		\begin{align}\label{eqn:h}
			h &=\frac{t}{(\delta-1)}\big(2(q-1) +\delta-1 -2\kappa (q-1)\big)-(\delta-q+1)
		\end{align}
		
		Note that the above value of $h$ has to be a positive integer, where $\kappa$ and $q$ are already given. Therefore, for $h$ to be a positive integer, the value of $t$ should be selected in such a way that $t>\frac{(\delta-1)(\delta-q+1)}{2(q-1) +\delta-1 -2\kappa (q-1)}$ and to satisfy a further requirement, we also need $t(q-1)^{2} \big(1+ \frac{2(\kappa-1)}{\delta-1}\big)$ to be a positive integer. It is clear that we can choose infinitely many positive integers $t$, satisfying the above property.
		
		\medskip
		
		Recall that the number of colors used in $G$ is at least $t(q-1)^{2}+ (h'+1)(q-1)+1$, where $h'=t-(\delta-q+2)-h$. By substituting the  value of $h$ from~\eqref{eqn:h} in the above expression, we then have that the number of colors used in $G$ to be at least, $t(q-1)^{2} + \Big[t-\delta +q -1 -\big[\frac{t}{\delta-1}\big(2(q-1) +\delta-1 -2\kappa (q-1)\big)-(\delta-q+1)\big]\Big](q-1) + 1$. On simplifying, we get the number of colors to be at least, 
		\medskip
		
		$t(q-1)^{2} \Big(1+ \frac{2(\kappa-1)}{\delta-1} + \frac{1}{t(q-1)^{2}}\Big) \geq t(q-1)^{2} \Big(1+ \frac{2(\kappa-1)}{\delta-1}\Big) = |E(H_{q-1})| \Big(1+ \frac{2(\kappa-1)}{\delta-1}\Big)$ as required in condition~(3) in the statement of the theorem.
		
		\medskip
		
		The above value is very close to our approximation, $|E(H_{q-1})| \Big(1+ \frac{2\big(\kappa+\frac{1}{q-1}\big)}{\delta-1}\Big)$.

	\end{proof}
	
	\begin{figure}
		\centering
		\begin{tikzpicture}
			\draw (0,0) ellipse(1cm and .5cm);
			\draw (3,0) ellipse(1cm and .5cm);
			\draw (7,0) ellipse(1cm and .5cm);
			\draw[thick,dotted] (4.8,0)--(5.3,0);
			\draw (11.9,0) ellipse(3.2cm and 1cm);
			\draw (1.5,4.5) ellipse(2.2cm and .9cm);
			\draw (6.5,4.5) ellipse(2.2cm and .9cm);
			\draw (12.5,4.5) ellipse(2.2cm and .9cm);
			\draw[thick,dotted] (9.3,4.5)--(9.8,4.5);
			\draw[thick] (.5,5.3)--(.5,3.7);
			\draw[thick] (1.7,5.4)--(1.7,3.6);
			\draw[thick] (5.5,5.3)--(5.5,3.7);
			\draw[thick] (6.7,5.4)--(6.7,3.6);
			\draw[thick] (11.5,5.3)--(11.5,3.7);
			\draw[thick] (12.7,5.4)--(12.7,3.6);
			\draw (2, 4.5) node {$\bullet$};
			\draw (2.5, 4.5) node {$\bullet$};
			\draw (3.5, 4.5) node {$\bullet$};
			\draw (7, 4.5) node {$\bullet$};
			\draw (7.5, 4.5) node {$\bullet$};
			\draw (8.5, 4.5) node {$\bullet$};
			\draw (13, 4.5) node {$\bullet$};
			\draw (13.5, 4.5) node {$\bullet$};
			\draw (14.5, 4.5) node {$\bullet$};
			\draw[thick] (9.5,0) circle(0.55cm);
			\draw[thick] (10.8,0) circle(0.55cm);
			\draw[thick] (12.1,0) circle(0.55cm);
			\draw[thick] (14.2,0) circle(0.55cm);
			\draw[thick,dotted] (12.9,0)--(13.4,0);
			
			\draw[thick,dashed] (0,4.4)-- (0,0.5);
			\draw[thick,dashed] (0,4.4)-- (12,1);
			
			\draw[thick,dashed] (5.1,4.4)-- (3,0.5);
			\draw[thick,dashed] (5.1,4.4)-- (12,1);
			
			\draw[thick,dashed] (11,4.4)-- (7,0.5);
			\draw[thick,dashed] (11,4.4)-- (12,1);
			
			\draw[thick] (1.2,4.4)-- (9.5,0.55);
			\draw[thick] (6.3,4.4)-- (9.5,0.55);
			\draw[thick] (12.2,4.4)-- (9.5,0.55);
			
			\draw[thick] (2,4.5)-- (10.8,0.55);
			\draw[thick] (7,4.5)-- (10.8,0.55);
			\draw[thick] (13,4.5)-- (10.8,0.55);
			
			\draw[thick] (2.5,4.5)-- (12.1,0.55);
			\draw[thick] (7.5,4.5)-- (12.1,0.55);
			\draw[thick] (13.5,4.5)-- (12.1,0.55);
			
			\draw[thick] (3.5,4.5)-- (14.2,0.55);
			\draw[thick] (8.5,4.5)-- (14.2,0.55);
			\draw[thick] (14.5,4.5)-- (14.2,0.55);
			
			
			\draw (0,4.75) node {$X_1$};
			\draw (5.1,4.75) node {$X_2$};
			\draw (11,4.75) node {$X_{q-1}$};
			
			\draw (1.2,4.75) node {$Y_1$};
			\draw (6.2,4.75) node {$Y_2$};
			\draw (12.2,4.75) node {$Y_{q-1}$};
			
			\draw (2.4,4.75) node {$Z_1$};
			\draw (7.5,4.75) node {$Z_2$};
			\draw (13.4,4.75) node {$Z_{q-1}$};
			
			\draw (9.5,0) node {$B_1$};
			\draw (10.8,0) node {$B_2$};
			\draw (12.1,0) node {$B_3$};
			\draw (14.2,0) node {$B_{h'+1}$};
			\draw (1.5,5.4) node[above] {$S_1$};
			\draw (6.5,5.4) node[above] {$S_2$}; 
			\draw (12.5,5.4) node[above] {$S_{q-1}$};
			\draw (0,-.6) node[below] {$T_1$};
			\draw (3,-.6) node[below] {$T_2$};
			\draw (7,-.6) node[below] {$T_{q-1}$};
			\draw (11.9,-1) node[below] {$B$};
		\end{tikzpicture}
		\caption{The figure illustrates the edges other than the edges belonging to the $(q-1)$-regular subgraph of the tight example
			in Theorem~\ref{thm:tightex1}. An ellipse represents a group of vertices,
			the number at the center is the cardinality of the group. A $\bullet$
			represents a single vertex. An edge (dashed or bold) between two groups indicates the
			complete set of edges between them.}
		\label{fig:fig4}
	\end{figure}
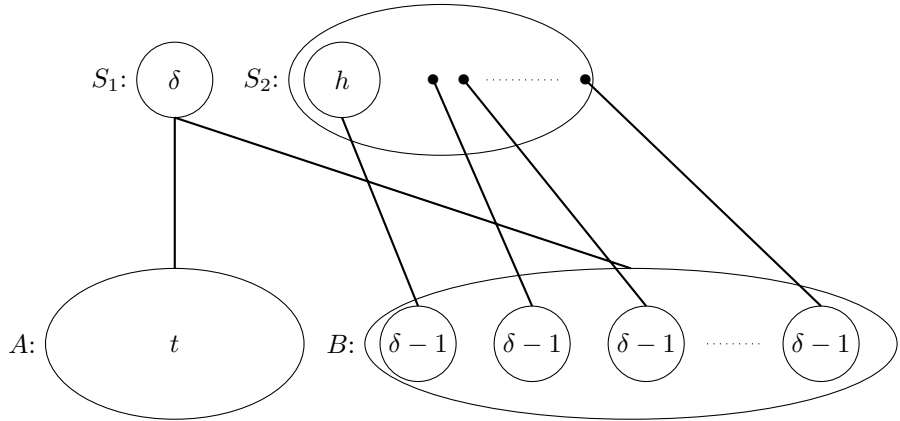

\end{document}